\documentclass[11pt]{article}
\usepackage[utf8]{inputenc}
\usepackage{tikz}
\usetikzlibrary{decorations.pathreplacing,calligraphy}
\usepackage{nicefrac}
\usepackage{amsmath}
\usepackage{amsfonts}
\usepackage{amssymb}
\usepackage{multicol}
\usepackage[T1]{fontenc}  %
\usepackage{marginnote}
\usepackage{amsthm}      %
\usepackage[noabbrev,capitalize]{cleveref}   %
\usepackage{caption}   %
\usepackage{subcaption}    %
\usepackage{booktabs}   %
\usepackage{multirow}   %
\usepackage[linesnumbered,lined,boxed,commentsnumbered]{algorithm2e}
\usepackage{wrapfig}
\usepackage{color}
\usepackage{adjustbox}

\usepackage{fullpage}

\usepackage[numbers]{natbib}   %

\begin{document}

\title{Robustness of Approval-Based Multiwinner Voting
  Rules\thanks{This paper is based on two conference
    papers~\citep{gaw-fal:c:approval-robustness,fal-gaw-kus:c:greedy-approval-robustness}. The
    paper streamlines the presentation of the results, includes the
    missing proofs, and contains additional discussion. On the other
    hand, some results---including all the experimental ones---were
    omitted.}}
  \author{Piotr Faliszewski\\
    AGH University\\
    Krakow, Poland \and
    Grzegorz Gawron\\
    AGH University and\\ VirtusLab\\
    Krakow, Poland \and
    Bartosz Kusek\\
    AGH University\\
    Krakow, Poland}

\maketitle

\newcommand{\calS}{\mathcal{S}}
\newcommand{\calT}{\mathcal{T}}
\newcommand{\ggcaps}[1]{\textsc{#1}}
\newcommand{\ggR}{\ensuremath{\mathcal{R}}}
\newcommand{\ggnote}[1]{\marginpar{#1}\reversemarginpar}
\newcommand{\ggfnote}[2]{\footnote{[#1]: #2}}
\newcommand{\ggrobustness}[2]{\ggcaps{\emph{#1}-#2-Robustness}}
\newcommand{\ggrobust}[2]{\ggcaps{\emph{#1}-#2-Robust}}
\newcommand{\ggx}{\circ}
\newcommand{\ggX}{\bigcirc}
\newcommand{\ggrr}[1]{\ggcaps{\emph{#1}-Robustness-Radius}}
\newcommand{\ggnr}[1]{\ggcaps{\emph{#1}-Noise-Robustness}}
\newcommand{\ggmath}[1]{\ensuremath{#1}}
\newcommand{\gggpav}{gPAV}
\newcommand{\gggccav}{gCCAV}
\newcommand{\ggpav}{PAV}
\newcommand{\ggccav}{CC}
\newcommand{\ggsav}{SAV}
\newcommand{\ggav}{{\ensuremath{\mathrm{AV}}}}
\newcommand{\gv}{\ggmath{v}}
\newcommand{\gvv}{\ggmath{w}}
\newcommand{\gvvv}{\ggmath{z}}
\newcommand{\gV}{\ggmath{V}}
\newcommand{\gVV}{\ggmath{W}}
\newcommand{\gVVV}{\ggmath{Z}}
\newcommand{\gc}{\ggmath{a}}
\newcommand{\gcc}{\ggmath{b}}
\newcommand{\gC}{\ggmath{A}}
\newcommand{\gCC}{\ggmath{B}}
\newcommand{\ggop}{\textsc{Op}}
\newcommand{\ggadd}{{\textsc{Add}}}
\newcommand{\ggremove}{\textsc{Remove}}
\newcommand{\ggswap}{\textsc{Swap}}
\newcommand{\ggaccept}{\textsc{Accept}}
\newcommand{\ggreject}{\textsc{Reject}}
\newcommand{\ggheadline}[1]{\smallskip\noindent\textbf{#1.}}
\newcommand{\ggwav}{\ggmath{\omega}-\ggav}

\newcommand{\ggfpt}{\ggmath{\mathrm{FPT}}}
\newcommand{\p}{{{\mathrm{P}}}}
\newcommand{\fp}{{{\mathrm{FP}}}}
\newcommand{\np}{{{\mathrm{NP}}}}
\newcommand{\fpt}{{{\mathrm{FPT}}}}
\newcommand{\sharpp}{{{\mathrm{\#P}}}}
\newcommand{\wone}{{{\mathrm{W[1]}}}}
\newcommand{\wtwo}{{{\mathrm{W[2]}}}}

\newcommand{\score}{{\mathrm{score}}}
\newcommand{\approval}{{\mathrm{app}}}
\newcommand{\av}{{\mathrm{AV}}}
\newcommand{\pav}{{\mathrm{PAV}}}
\newcommand{\ccav}{{\mathrm{CC}}}
\newcommand{\sav}{{\mathrm{SAV}}}

\newcommand{\calR}{\mathcal{R}}
\newcommand{\pref}{\succ}

\newcommand{\ggTODO}{{\textcolor{red}{TODO}}}

\newtheorem{theorem}{Theorem}[section]
\newtheorem{lemma}[theorem]{Lemma}
\newtheorem{proposition}[theorem]{Proposition}
\newtheorem{definition}{Definition}[section]
\newtheorem{corollary}[theorem]{Corollary}

\begin{abstract}
  We investigate how robust approval-based multiwinner voting rules
  are to small perturbations in the votes.  In particular, we consider
  the extent to which a committee can change after we add/remove/swap
  one approval, and we consider the computational complexity of
  deciding how many such operations are necessary to change the set of
  winning committees. We also consider the counting variants of our
  problems, which can be interpreted as computing the probability that
  the result of an election changes after a given number of random
  perturbations of the given election.
\end{abstract}

\section{Introduction}

The goal of a multiwinner election is to select a committee, i.e., a
fixed-size set of candidates, based on the opinions of the voters.
For example, citizens of a country may choose the members of their
parliament, judges in a competition may select a group of its
finalists, and a company may choose the products to offer based on the
preferences of its customers; naturally, for each of these
applications we would use a different voting rule, with different
properties. Unfortunately, even the most appropriate rule may give
unsatisfying results if the input votes are distorted.  Indeed, some
votes may be recorded erroneously due to mistakes of the voters or due
to the mistakes of the election officials (or their machinery). In
either case, it is interesting to know the consequences of such
distortions.  To address this issue, Bredereck et
al.~\cite{bre-fal-kac-nie-sko-tal:j:robustness} initiated the study of
robustness of multiwinner elections to small changes in the input
votes. We follow up on their ideas, but instead of considering ordinal
elections, we focus on the approval-based ones. (In the approval
setting each voter simply indicates which candidates he or she finds
acceptable; in the ordinal case, the voters rank the candidates from
the most to the least appealing one.)

\subsection{Robustness in Multiwinner Elections}
We are interested both in the extent to which a winning committee can
change---in the worst case---due to a small perturbation of a single
vote, and in computing the number of such perturbations necessary to
change the result of a specific, given election. Additionally, we are
also interested in computing the probability that introducing a given
number of (randomly selected) perturbations changes the result.

Regarding the first issue, we use the notion of the \emph{robustness
  level} of a multiwinner rule. Bredereck at
al.~\cite{bre-fal-kac-nie-sko-tal:j:robustness} defined it to be
$\ell$ if it is guaranteed that after swapping a pair of adjacent
candidates in a vote (in the ordinal election model) the winning
committee remains the same, except that up to $\ell$ candidates may be
replaced and, indeed, there are cases when this happens. In the
approval setting, instead of swapping adjacent candidates we consider:
\begin{enumerate}
\item adding a single approval for a single candidate,
\item removing a single approval from a single candidate, or
\item swapping a single approval in a single vote (i.e., moving it
  from one candidate to the other).
\end{enumerate}
In consequence, we consider three robustness level notions, one for
each way of modifying the election.

Regarding the number of perturbations needed to change the election
result, we consider three variants of the \textsc{Robustness Radius}
problem, originally introduced by Bredereck et
al.~\cite{bre-fal-kac-nie-sko-tal:j:robustness}, where we ask how many
approvals have to be added, removed, or swapped to change the election
result. To compute the probability that a number of randomly
introduced perturbations change the election result, we consider
counting versions of these problems (see below).

\subsection{Multiwinner Voting Rules}
We obtain our results by considering the following seven voting rules:
Approval Voting (AV), Satisfaction Approval Voting
(SAV~\cite{bra-kil:b:sav}), approval-based Chamberlin--Courant
(CC~\cite{cha-cou:j:cc,pro-ros-zoh:j:proportional-representation}),
Proportional Approval Voting (PAV~\cite{thi:j:pav}), greedy
variants of the latter two rules (GreedyCC and
GreedyPAV~\citep{bou-lu:c:chamberlin-courant,azi-bri-con-elk-fre-wal:cOUTDATED:justified-representation}),
as well as Phragm\'en's sequential rule
(Phragm\'en~\citep{jan:t:phragmen,bri-fre-jan-lac:c:phragmen}).  All
these rules are different in their nature and capture different types
of multiwinner elections. For example, AV chooses individually
excellent candidates and can be used to select finalists of
competitions, CC chooses diverse committees and can be used to
select products that a company should offer to its
customers~\cite{elk-fal-sko-sli:j:multiwinner-properties,fal-sko-sli-tal-tal:j:hierarchy-committee},
PAV aims at representing the voters
proportionally~\cite{azi-bri-con-elk-fre-wal:j:justified-representation,lac-sko:c:approval-thiele}
and Phragm\'en captures a different approach to proportionality (see
the textbook of \citet{lac-sko:b:multiwinner-approval} for a detailed discussion of
proportionality in approval elections; see also the
work of \citet{bri-las-sko:j:apportionment}).  GreedyCC and GreedyPAV
are approximate, sequential vairants of the CC and PAV rules,
respectively, and as opposed to these rules, they
are polynomial-time computable. SAV is closest in spirit to AV, but
has its own quirks. We refer to the chapters of
Kilgour~\cite{kil:chapter:approval-multiwinner}, Faliszewski et
al.~\cite{fal-sko-sli-tal:b:multiwinner-voting}, \citet{bau-fal-rot-sko:b:multiwinner}, and the textbook of
\citet{lac-sko:b:multiwinner-approval} for more details on multiwinner
voting.

\subsection{Our Results}\label{sec:summary}

We summarize our results in Table~\ref{tab:results}.  As in the work
of Bredereck et al.~\cite{bre-fal-kac-nie-sko-tal:j:robustness}, we
find that all our rules either have robustness level $1$ (so that a
single small change to the voters' preferences leads to replacing at
most one candidate in the winning committee) or $k$ (so that a single
change can lead to replacing the whole committee). Yet, we find two
interesting quirks regarding the SAV rule. First, we find that even
though the rule is polynomial-time computable, its robustness
level---for the case of adding or deleting approvals---is $k$; this is
the first example of such a simple rule with this behavior.  While we
also show that the same holds for other polynomial-time computable
rules---namely GreedyCC, GreedyPAV, and Phragm\'en---here the
situation is particularly interesting because SAV is based on simply
summing up the points assigned to the candidates (whereas the other
just-mentioned rules are notably more involved).  The second quirk
regarding SAV is that even though its robustness level for adding and
removing approvals is $k$, for the case of swapping approvals it is
$1$.  On the technical side, this result is a simple consequence of
the rule's definition (see Sections~\ref{sec:prelim}
and~\ref{sec:robustness} for details), but it is interesting to have
an example of such behavior. It also motivates studying the three ways
of perturbing the votes separately, even though in many cases the results
and proofs are very similar.

\begin{table}[t]
  \centering
  \caption{Summary of our results. In the three
    columns marked ``Robustness Level'' we provide the robustness
    level values for our rules; in the columns marked ``Robustness Radius''
    we indicated the complexity results for the \textsc{Robustness Radius} problem
    (including its counting variant for the AV and SAV rules; these results
    regard adding and removing approvals only, but not swapping
    them). We also indicate which of the rules are polynomial-time
    computable and which are $\np$-hard (but these results are not due
    to this paper).}
\label{tab:results}

\scalebox{0.9}{
\begin{tabular}{c|c|ccc|ccc}
  \toprule
                \multicolumn{2}{c|}{} &\multicolumn{3}{c|}{} & \multicolumn{3}{c}{Robustness Radius}\\
                \multicolumn{2}{c|}{} &\multicolumn{3}{c|}{Robustness Level} & \multicolumn{3}{c}{(decision/counting)}\\
\midrule
        &winner  & Adding  & Removing & Swapping  &  Adding  & Removing & Swapping  \\
rule      &det. & Approvals &  Approvals &  Approvals   &  Approvals &  Approvals &  Approvals \\ \midrule
AV     & $\p$ & $1$     & $1$   & $1$  &  $\p$/$\fp$ & $\p$/$\fp$ & $\p$/---  \\ 
SAV    & $\p$ & $k$     & $k$   & $1$  &  $\p$/$\sharpp$-com. &  $\p$/$\sharpp$-com. &  $\p$/---  \\ 
\midrule
CC       & $\np$-hard& $k$     & $k$   & $k$  &  $\np$-hard &  $\np$-hard &  $\np$-hard \\ 
PAV    & $\np$-hard& $k$     & $k$   & $k$  &  $\np$-hard &  $\np$-hard &  $\np$-hard \\
\midrule
GreedyCC  & P & $k$     & $k$   & $k$  &  $\np$-com. &  $\np$-com. &  $\np$-com. \\
GreedyPAV   & P & $k$     & $k$   & $k$  &  $\np$-com. &  $\np$-com. &  $\np$-com. \\
Phragm{\'e}n &   P & $k$     & $k$   & $k$  &  $\np$-com. &  $\np$-com. &  $\np$-com. \\

\bottomrule               
\end{tabular}
}

\end{table}

Regarding the \textsc{Robustness Radius} problem, we find that its
decision variants
are polynomial-time computable for AV and SAV,
but become $\np$-hard for all the other rules.\footnote{In the
  conference version of one of the papers on which this one is
  based~\citep{gaw-fal:c:approval-robustness} we mention FPT
  algorithms for CC and PAV parameterized by the number of candidates
  or the number of voters. We omit these results here to streamline
  the narrative. The idea behind the proof is very similar as that for
  bribery FPT algorithms of
  \citet{fal-sko-tal:c:bribery-success}. Informally put, an FPT
  algorithm is an algorithm that runs in polynomial time with respect
  to the size of the input and, possibly, in superpolynomial time with
  respect to the so-called ``parameter.'' For example, an FPT
  algorithm parameterized by the number of candidates can have running
  time of the form $O(2^m\cdot mn)$, where $m$ is the number of
  candidates and $n$ is the number of voters. For more details on
  parameterized algorithms, see the textbooks of
  \citet{nie:b:invitation-fpt} and
  \citet{cyg-fom-kow-lok-mar-pil-pil-sau:b:fpt}.}  The reader may
wonder if this $\np$-hardness can be extended to
$\np$-completeness. While for polynomial-time computable, sequential
rules, this indeed is the case (as reported in \Cref{tab:results}),
the situation of CC and PAV is unclear. The reason is that for these
rules it seems intractable to check if after adding/removing/swapping
approvals the result changed, which precludes a natural
$\np$-membership argument.

For AV and SAV
we also consider the counting variants of the \textsc{Robustness Radius} problem (for adding and
removing approvals; we decided not to consider approval swapping for
the sake of simplicity). The idea is as follows (we focus on adding
approvals here): Let $X$ be the number of ways in which it is possible
to add $B$ approvals to an election without changing the result, and
let $Y$ be the total number of ways in which it is possible to add $B$
approvals. The fraction $\nicefrac{X}{Y}$ gives the probability of
\emph{not changing the election result} by adding $B$ approvals
uniformly at random.
By computing this probability, we can distinguish situations where
perturbing the election \emph{may} possibly affect the election
results but it is very unlikely, from those where we expect the result
to be changed.
It turns out that our counting problems are in $\fp$ for the AV rule,
but are $\sharpp$-hard for the SAV rule.\footnote{The class $\fp$
  contains polynomial-time computable functions and the class
  $\sharpp$ is the analogue of $\np$ for counting problems.}

\subsection{Related Work}

Our work is most closely related to that of Bredereck et
al.~\cite{bre-fal-kac-nie-sko-tal:j:robustness}, which we have already
discussed above. Three other very closely related papers are due to
Misra and Sonar~\cite{mis-son:c:cc-robustness-sp}, Faliszewski et
al.~\cite{fal-sko-tal:c:bribery-success}, and
\citet{boe-fal-jan-kac:c:pb-robustness}.  Misra and
Sonar~\cite{mis-son:c:cc-robustness-sp} consider the
\textsc{Robustness Radius} problem for the case of the
Chamberlin--Courant rule, for some variants of single-peaked and
single-crossing elections and for the approval setting.  Their result
for the case of approval-based Chamberlin--Courant rule is similar to
ours and, to some extent, is stronger (in particular, they require
fewer approved candidates per voter and, after dropping this
restriction, they also obtain $\wtwo$-hardness for the parametrization
by the committee size), but the advantage of our result is that in a
single proof we cover a large subfamily of Thiele
rules~\cite{sko-fal-lan:j:collective,lac-sko:c:approval-thiele},
including CC and PAV.

Faliszewski et al.~\cite{fal-sko-tal:c:bribery-success} study the
problem of bribery in approval-based elections, where the goal is to
ensure that a particular, specified candidate becomes a member of the
winning committee by either adding, removing, or swapping
approvals. Our work is similar in that we use the same types of
operations, but our problems focus on changing the set of winning
committees and not on ensuring some candidate's victory.
\citet{kus-bre-fal-kac-kno:c:structured-approval-bribery} recently
followed-up on the work of \citet{fal-sko-tal:c:bribery-success} by
considering their problems in structured domains, also looking at
destructive variants, which are close in spirit to the problems that
we consider (for an overview of structured domains in approval
elections, see the works of Elkind et
al.~\citep{elk-lac-pet:b:structured,elk-lac:approval-ci-vi}). They
only considered the AV rule and, somewhat surprisingly, found many of
their problems to be intractable.

Finally, the work of \citet{boe-fal-jan-kac:c:pb-robustness} is a
follow-up to the conference papers on which this one is
based. \citet{boe-fal-jan-kac:c:pb-robustness} focus on robustness in
participatory budgeting elections, giving both theoretical and
experimental results. The main difference between the participatory
budgeting setting and the multiwinner one is that while in the latter
we seek a committee of a fixed size, in the former each candidate
(project) has a possibly different price and we seek a committee of at
most a given total cost (see the survey of participatory budgeting
issues in computational social choice due to
\citet{rey-mal:t:pb-survey}).  \citet{boe-fal-jan-kac:c:pb-robustness}
consider a variant of the \textsc{Robustness Radius} problem where it
is possible to both add and remove approvals (they refer to it as
\textsc{Flip Bribery}) and their complexity analysis largely focuses
on the participatory budgeting variant of the AV rule (although they
also consider other rules).

Another paper that is very related to ours on the technical level is that of
\citet{jan-fal:c:ties-approval-voting}. There, the authors study the
problem of detecting ties in approval-based multiwinner voting. While
their motivation is very different, their proofs for GreedyCC and
GreedyPAV (or, more generally, greedy unit-decreasing Thiele rules;
see \Cref{sec:prelim}) are, in essence, extensions of ours.

Prior to the work of Bredereck et
al.~\cite{bre-fal-kac-nie-sko-tal:j:robustness}, Shiryaev et
al.~\cite{shi-yu-elk:c:robust-winners} studied the robustness of
single-winner elections. Specifically, they asked for the complexity
of the \textsc{Destructive Swap Bribery} problem, where the goal is to
change the winner of an election by making as few swaps of adjacent
candidates as possible (they considered the ordinal
setting). Kaczmarczyk and Faliszewski~\cite{kac-fal:j:dsb} also
studied a variant of this problem, where each swap involves a
specified candidate (the original winner of the
election).\footnote{This problem is known as \textsc{Destructive Shift
    Bribery}. The constructive variant was introduced by Elkind et
  al.~\cite{elk-fal-sli:c:swap-bribery}. \textsc{Shift Bribery} was
  studied deeply in a number of
  papers~\cite{bre-che-fal-nic-nie:j:prices-matter,bre-fal-nie-tal:j:committee-shift-bribery,zho-guo:c:parameterized-iterative-shift-bribery,fal-man-sor:j:approx-shift-bribery,mau-nev-rot-sel:j:iterative-shift-bribery}.}
What makes our problem different is the focus on multiwinner approval
elections (rather than single-winner, ordinal ones) and looking for
any change to the winning committee set (rather than having a specific
candidate win).  Other \textsc{Destructive Bribery}
problems~\cite{fal-hem-hem-rot:j:llull} were studied under the name
\textsc{Margin of
  Victory}~\cite{mag-riv-she-wag:c:stv-bribery,car:c:stv-bribery,xia:c:margin-of-victory,dey-nar:c:margin-of-victory}
and focused on finding the smallest number of votes that need to be
modified to change the election result (so, as opposed to the previous
problems, in the \textsc{Margin of Victory} problems only the number
of modified votes matters, irrespective of the extent to which they
are changed).  The idea is that if modifying a few votes can lead to
changing the electon result, then it is possible that the election was
manipulated and a detailed audit may be necessary.

Regarding the counting variant of the \textsc{Robustness Radius}
problem, so far relatively few authors considered counting variants of
election-related problems. Some notable examples of works that do
study counting problems include those of Hazon et
al.~\cite{haz-aum-kra-woo:j:uncertain-election-outcomes}, Wojtas et
al.~\cite{fal-woj:c:counting-control}, and
\citet{boe-bre-fal-nie:c:counting-swap-bribery}, all of which focus on
single-winner elections (however these certainly are not the only
examples of papers studying counting problems on elections; for
example, \citet{boe-fal-jan-kac:c:pb-robustness} also consider
counting variants of their robustness problems). Hazon et
al.~\cite{haz-aum-kra-woo:j:uncertain-election-outcomes} considered
the problem of computing the probability that a given candidate is an
election winner, provided that for each voter there is a probability
distribution over the votes he or she may cast. Wojtas et
al.~\cite{fal-woj:c:counting-control} considered the problem of
computing the probability that a given candidate is a winner, provided
that a given number of candidates or voters is added/removed from the
election. \citet{boe-bre-fal-nie:c:counting-swap-bribery} studied
robustness of single-winner voting rules for ordinal elections to
swapping pairs of candidates in the votes.

To conclude, we mention that election robustness is a broad term that
is studied in other contexts as well, e.g., in electronic voting and
political science, but we omit the discussion of this literature as it
is quite distant from our work.

\section{Preliminaries}\label{sec:prelim}
We assume basic familiarity with computational complexity theory and
we point to the textbook of Papadimitriou~\cite{pap:b:complexity} for
more details on this topic. We mention that $\sharpp$ is a class of
counting problems, analogous to $\np$. The main feature is that
instead of asking if some computation path accepts, the answer is the
number of such accepting paths.  For an integer $t$, we write $[t]$ to
mean the set $\{1, \ldots, t\}$.

\subsection{Elections and Rules}
An \emph{approval-based election} $E=(C,V)$ consists of a set of
candidates $C = \{c_1, \ldots, c_m\}$ and a collection of voters
$V = (v_1, \ldots, v_n)$, where each voter $v_i \in V$ approves a
subset $A(v_i)$ of candidates from $C$ (so $|A(v_i)|$ is the number of candidates that $v_i$ approves).
For a candidate $c_j$, we write $V(c_j)$ to mean
the set of voters that approve $c_j$, and we write $\approval_E(c_j)$ to
mean the cardinality of this set. We sometimes refer to
$\approval_E(c_j)$ as the approval score of candidate $c_j$.

An \emph{approval-based multiwinner voting rule} \ggR{} is a function
that given an approval-based election and an integer $k \in [|C|]$
provides a family $\ggR{}(E, k)$ of size-$k$ committees, i.e., a
family of size-$k$ subsets of $C$, that tie for victory.  Some of our
rules---namely those that proceed by adding candidates to the
committee sequentially---will always output a single winning
committee, wheres others---namely those that associate committees with
scores---will sometimes output more winning committees.  In practical
applications, the latter rules need some sort of a tie-breaking
mechanism, but we omit this issue in our considerations.

Next, we describe our score-based rules. Let $E = (C,V)$ be an
election and fix committee size~$k$.  For each of the rules described
below, we define the score that it assigns to a given committee $S$ and
the rule selects all the committees with the highest score:
\begin{description}

\item[Approval Voting (AV).] Under the AV rule, the score of committee
  $S$ is defined as the sum of the approval scores of its members;
  formally we have:
  \[\textstyle \score_E^{\av}(S) = \sum_{c \in S} \approval_E(c).\]

\item[Satisfaction Approval Voting (SAV).] Under the SAV rule,
  introduced by Brams and Kilgour~\cite{bra-kil:b:sav}, each voter has
  a single point which he or she distributes among all the candidates
  that he or she approves.  For a candidate $c$, we refer to the value
  $\sum_{v \in V(c)}\frac{1}{|A(v)|}$ as the SAV score of $c$.  The SAV
  score of committee $S$ is the sum of the SAV scores of its members:
  \[
    \textstyle \score_E^{\sav}(S) = \sum_{c \in S}\left( \sum_{v \in V(c)}\frac{1}{|A(v)|} \right).
  \]
 
\item[Thiele Rules (Including CC and PAV).] A Thiele rule for the
  case of size-$k$ committees is defined through a vector
  $\omega = (\omega_1, \ldots, \omega_k)$ of non-negative real
  numbers. The score of committee $S$ under a Thiele rule specified by
  such a vector is defined as:
  \[
    \textstyle \score_E^{\omega\hbox{-}\av}(S) = \sum_{v \in V}\left( \sum_{i=1}^{|A(v) \cap S|} \omega_i \right).
  \]
  Examples of Thiele rules include the AV rule, defined using vectors
  of the form $\omega^\av = (1, \ldots, 1)$, the approval-based
  Chamberlin--Courant rule
  (CC)~\cite{cha-cou:j:cc,pro-ros-zoh:j:proportional-representation},
  defined using vectors of the form $\omega^\ccav = (1, 0, \ldots, 0)$,
  and the Proportional Approval Voting rule (PAV), defined using
  vectors of the form
  $\omega^\pav = (1, \nicefrac{1}{2}, \ldots, \nicefrac{1}{k})$.  This
  family of rules was introduced by Thiele~\cite{thi:j:pav}, who in
  particular introduced the PAV rule.
  We are mostly interested in Thiele rules with polynomial-time
  computable vectors
  $\omega = (\omega_1, \omega_2, \ldots, \omega_k)$, such that
  $1 = \omega_1 > \omega_2 \geq \omega_3 \geq \cdots \geq
  \omega_k$. We refer to such rules as \emph{unit-decreasing Thiele
    rules}. CC and PAV are unit-decreasing, but AV is not.  
\end{description}

Let us discuss CC, PAV, and Thiele rules in some more detail.
Intuitively, under the CC rule a voter assigns score $1$ to a
committee if he or she approves at least one member of this committee,
and assigns score $0$ otherwise (so a voter is satisfied with a
committee if it contains a candidate that the voter can see as his or
her \emph{representative}). Under the PAV rule, the appreciation of a
voter for a committee increases with the number of its members that he
or she approves, but the more members a voter approves, the smaller is
the marginal value of approving further ones.
The form of the vector used by PAV makes the rule suitable for
parliamentary
elections~\cite{azi-bri-con-elk-fre-wal:j:justified-representation}
and, indeed, it can be seen as a generalization of the D'Hondt
apportionment method used in parliamentary elections in many
countries~\cite{bri-las-sko:j:apportionment}.

There are polynomial-time algorithms for computing the winning
committees under the AV and SAV
rules~\cite{azi-gas-gud-mac-mat-wal:c:approval-multiwinner}, but the
problem of deciding if there is a committee with at least a given
score is $\np$-hard for
CC~\cite{pro-ros-zoh:j:proportional-representation,bou-lu:c:chamberlin-courant},
PAV~\cite{sko-fal-lan:j:collective,azi-gas-gud-mac-mat-wal:c:approval-multiwinner},
and many other Thiele rules (but there is a number of algorithmic
workarounds, including approximation
algorithms~\cite{sko-fal-sli:j:multiwinner,byr-sko-sor:t:pav-approx,dud-man-mar-sor:c:thiele-pav-approx}
and $\fpt$
algorithms~\cite{bet-sli-uhl:j:mon-cc,fal-sko-sli-tal:j:top-k-counting}),
and polynomial-time algorithms for restricted
domains~\cite{bet-sli-uhl:j:mon-cc,sko-yu-fal-elk:j:mwsc,pet-lac:j:spoc,sor-wil-xu:c:thiele-structured}.

In addition to CC and PAV, we are also interested in their sequential
variants, GreedyCC and GreedyPAV, defined by their greedy
algorithms. Below we define the general form of these algorithms, for a given election $E= (C,V)$,
committee size $k$, and a vector $\omega = (\omega_1, \ldots, \omega_k)$:
\begin{description}
\item[Greedy Thiele Rules (Including GreedyCC and GreedyPAV).] We
  start with an empty committee $W = \emptyset$ and perform $k$
  iterations, where in each iteration we extend $W$ with a single
  candidate $c$ that %
  maximizes the value
  $\score_E^{\omega\hbox{-}\av}(W \cup \{c\}) -
  \score_E^{\omega\hbox{-}\av}(W)$.  If several candidates satisfy
  this condition then we break the tie according to a given
  tie-breaking order defined over the candidates (and which is part of
  the input). We output $W$ as the unique winning committee.
  Analogously to the case of Thiele rules, we also speak of
  unit-decreasing greedy Thiele rules when the $\omega$ vectors
  satisfy $1 = \omega_1 > \omega_2 \geq \cdots \geq \omega_k$ and are
  polynomial-time computable.
\end{description}
We refer to the Greedy Thiele rules defined by $\omega^\ccav$ and
$\omega^\pav$ as GreedyCC and GreedyPAV, respectively.
When analyzing an $i$-th iteration of these algorithms, for each
candidate $c$ we refer to the value
$\score_E^{\omega\hbox{-}\av}(W \cup \{c\}) - \score_E^{\omega\hbox{-}\av}(W)$ as the
score of $c$.  For GreedyCC, we imagine that as soon as a candidate is
included in the committee, all the voters that approve him or her are
removed (indeed, these voters would not contribute positive score to any
further candidates).

We are also interested in the Phragm{\'e}n rule (or, more specifically,
in Phragm{\'e}n's sequential rule, but we use the shorter name in this paper). The Phragm{\'e}n rule proceeds according to the
following scheme ($E = (C,V)$ is the input election and~$k$ is the
committee size):
\begin{description}
\item[Phragm\'en.] Initially, we have committee $W = \emptyset$. The
  voters start with no money, but they receive it at a constant rate
  (so, at each time point $t \in \mathbb{R}$, $t \geq 0$, each voter
  has in total received money of value $t$). At every time point for
  which there is a candidate $c$ not included in $W$ who is approved
  by voters that jointly have one unit of money, this candidate is
  ``purchased.''  That is, candidate~$c$ is added to $W$ and the
  voters that approve him or her have all their money reset to $0$
  (i.e., they pay for $c$). If several candidates can be purchased at
  the same time, we consider them in a given tie-breaking order (given
  as part of the input). The process continues until $W$ reaches size
  $k$ or all the remaining candidates have approval score zero (in
  which case we extend $W$ according to the tie-breaking order). We
  output $W$ as the unique winning committee.
\end{description}
Similarly to PAV, Phragm\'en provides committees that ensure
proportional representation of the
voters~\cite{san-elk-lac-fer-fis-bas-sko:c:pjr}.  For a detailed
discussion of all our rules (including an alternative definition of
Phragm{\'e}n), we point the reader to the survey of Lackner and
Skowron~\cite{lac-sko:b:multiwinner-approval}.  Faliszewski et
al.~\cite{fal-sko-sli-tal:b:multiwinner-voting} and
\citet{bau-fal-rot-sko:b:multiwinner} offer a general overview of
multiwinner voting.

\subsection{Robustness Notions}
Next, we formally introduce the notions of robustness level and
robustness radius for approval multiwinner voting rules, by adapting
the approach of Bredereck et
al.~\cite{bre-fal-kac-nie-sko-tal:j:robustness} from the world of
ordinal multiwinner elections.  There are three types of operations
on approvals that we are interested in.  By the \ggadd{} operation, we
mean adding an approval for a given candidate in a given vote, by the
\ggremove{} operation we mean removing an approval from some candidate
in a given vote, and by the \ggswap{} operation we mean moving an
approval from one candidate to another in a given vote.

For each operation $\mathrm{Op} \in \{\ggadd,\ggremove,\ggswap\}$ we
say that a multiwinner rule $\ggR$ is $\ell$-{Op-Robust} if performing
a single operation of type {Op}, in a given election, leads to
replacing at most $\ell$ members of the winning committee. Formally,
we have the following definition (the robustness level notion used by
Bredereck et al.~\cite{bre-fal-kac-nie-sko-tal:j:robustness} is the
same, except that they considered ordinal elections and the operation
of swapping adjacent candidates in a preference order).

\begin{definition}\label{defRL}
  Let \ggR{} be a voting rule and let $\mathrm{Op} \in \{\ggadd{}$,
  $\ggremove{}$, $\ggswap{}\}$ be an operation type.  We say that the
  \emph{\ggop-robustness level} of \ggR{} is $\ell$ (\ggR{} is
  {$\ell$}-$\mathrm{Op}$-robust) if $\ell$ is the smallest number such that
  for each election $E = (C,V)$, each committee size $k$, $k\leq |C|$,
  and each election $E'$ obtained from $E$ by applying a single
  $\mathrm{Op}$ operation the following holds:
  For each committee $W \in \ggR{}(E,k)$ there exists a committee
  $W' \in \ggR{}(E',k)$ such that $|W \cap W'| \geq k - \ell$.
\end{definition}
By a slight abuse of notation, we will often speak of
$k$-{Op-Robustness} to mean that a single Op operation may lead to
replacing the whole committee.
We are also interested in Bredereck et
al.'s~\cite{bre-fal-kac-nie-sko-tal:j:robustness} \textsc{Robustness
  Radius} problem.%

\begin{definition}\label{defRR}
  Let \ggR{} be a voting rule and let
  $\ggop{} \in \{\ggadd{}, \ggremove{}, \ggswap{}\}$ be an operation
  type.  In the \ggrr{\ggR{}-\ggop{}} problem we are given an approval
  election $E=(C,V)$, a committee size $k$, an integer $B$, and we ask
  if by applying a sequence of $B$ operations of type $\ggop$ it
  is possible to obtain an election $E'$ such that
  $\ggR{}(E, k) \neq \ggR{}(E', k)$.
\end{definition}

In addition to the above problems, we also consider their counting
variants (for the $\ggadd$ and $\ggremove$ operations), where we
consider the number of ways in which $B$ operations of a given type
can be performed so that the result of the election \emph{does not
  change}.\footnote{Since it is easy to compute the total number of
  ways of performing $B$ operations of a given type, from the
  computational complexity point of view it is irrelevant if we count
  the cases where the result changes or does not change , but the
  latter approach simplifies our proofs.}

\subsection{Matrix Notation} In many of our proofs we represent
elections and operations performed on them as matrices, where each
entry is either $\ggx$, $+$, $-$, or an empty space. Each such matrix
specifies two elections, the original election $E$ and the election
$E'$, obtained from $E$ by applying some operations.  Each row
corresponds to a voter and each column corresponds to a candidate.  A
symbol in row $i$ and column $j$ indicates whether voter $v_i$
approves candidate $c_j$. Symbol~$\ggx$ indicates that a given
candidate is approved in both elections, symbol~$-$ indicates that the
candidate is approved in $E$ but not in $E'$, and symbol~$+$ indicates
\begin{figure}
  \centering
\arraycolsep=1.0pt\def\arraystretch{1.0}
$\begin{array}{c|ccc}
    & c_1 & c_2 & c_3 \\
    \hline
  v_1 & \ggx &  - &  + \\
  v_2 & \ggx &    &    \\
  \end{array}$
        \caption{The matrix notation for specifying elections and
          operations on them.}
        \label{fig-matrix_def}
\end{figure}
that the candidate is approved in $E'$ but not in $E$ (so we use $+$
for the \textsc{Add} operation, $-$ for the \textsc{Remove} operation,
and both a $+$ and a $-$ for the \textsc{Swap} operation).

We present an example matrix in Figure~\ref{fig-matrix_def}. The
elections that it specifies have three candidates, $c_1$, $c_2$ and
$c_3$, and two voters, $v_1$ and $v_2$. In the original election,
$v_1$ approves $c_1$ and $c_2$, and $v_2$ approves $c_1$. In the
election obtained after applying the specified operation (in this case
it is a \textsc{Swap} operation), $v_1$ approves $c_1$ and $c_3$, and
$v_2$'s vote remains unchanged.

\section{Robustness Levels}\label{sec:robustness}

In this section we analyze the robustness levels of our rules. We
first consider the AV rule, which is $1$-robust for each operation
type.

\begin{proposition}\label{propAV}
$\ggav$  is $\mathrm{1\hbox{-}Op\hbox{-}Robust}$ for each $\mathrm{Op} \in \{\ggadd{}, \ggremove{}, \ggswap{}\}$.
\end{proposition}

\begin{proof}
  Let $E$ be an \ggav{} election with some winning committee $W$ and
  $E'$ an election obtained from $E$ by adding an approval to some
  candidate $c$.  If $c$ is in $W$ then the committee $W$ is still
  winning in $E'$ (so there is no change to the winning committee).
  If $c$ is not in $W$ then let $d$ be some member of $W$ with the lowest
  AV score. Now, either the committee $W$ is still winning in the
  election $E'$ or committee $(W \setminus \{d\}) \cup \{c\}$ is
  winning.

  Similarly, in the \ggremove{} case, w.l.o.g., we let $W$ be some
  winning committee and $c$ be the candidate that has one of its
  approvals removed.  When $c \notin W$ the removal doesn't change the
  score of committee $W$, hence this committee is still winning.
  When, on the other hand, $c \in W$, then it might turn out that the
  score of the winning committee can be improved by replacing $c$ with
  a $d \notin W$ where $\approval_E(c) = \approval_E(d)$ (i.e., where
  $d$ has the same approval score as $c$ had prior to the removal).
  Any other member of $W$ will still be the member of the committee,
  because no other votes were changed.  Therefore \ggav{} is
  \ggrobust{1}{Remove}.

  Finally, in the \ggswap{} case, w.l.o.g., we let $W$ be some winning
  committee and let the swap operation remove approval from $c$ and
  add it to $d$.  If $c, d \in W$, then swap could cause at most one
  of them to be replaced in $W$ (approval score of $d$ increases, so
  it stays in $W$, but $c$ might get replaced).  If $c, d \notin W$,
  then swap could have at most one member of~$W$ be replaced by $d$
  (whose score increased by $1$).  If $c \notin W$ and $d \in W$, then
  the committee $W$ will still be winning, since it will have more
  approvals than before.  Lastly, if $c \in W$ and $d \notin W$ then,
  again, only one member of $W$ could be replaced. To see this, let
  $x$ denote the lowest scoring member of $W$ prior to the swap and
  note that if there are two new candidates to enter $W$, then both
  their scores must be greater than the score $x$, because all the
  candidates in $W$ (apart from possibly $c$) have their scores after
  the swap not lower than $x$.  But the only candidate not belonging
  to the winning committee that can have such a score greater than $x$
  after the swap is $d$.  Thus, we conclude that \ggav{} is
  \ggrobust{1}{Swap}.
\end{proof}

\begin{figure}
  \centering
\begin{subfigure}{0.45\textwidth}
  {\footnotesize 
  \centering
  \arraycolsep=1.0pt\def\arraystretch{1.0}
  $\begin{array}{c|ccccccc}
        & a_1 & a_2 & a_3 & b_1 & b_2 & b_3 \\
    \hline
    v_1 &\ggx &\ggx &\ggx &  +  &     &     \\
     v_2 &     &     &     &\ggx &\ggx &\ggx \\
     \multicolumn{1}{c}{} & & & & & &
   \end{array}$
  \caption{An election for adding an approval, for $k=3$.}
  \label{fig-SAV-add-RL}
      }
    \end{subfigure}\quad
\begin{subfigure}{0.45\textwidth}
	{\footnotesize 
	\centering
	\arraycolsep=1.0pt\def\arraystretch{1.0}
	$\begin{array}{c|ccccccccccccccc}
		    & s  & a_1 & a_2 & b_1 & c_1 & c_2 & c_3 & c_4 & c_5 & d_1 & d_2 & d_3 & d_4 & d_5 \\
		\hline
		v_1 & -  &\ggx &\ggx &     &     &     &     &     &     &     &     &     &     &     \\
		v_2 &    &     &     &\ggx &\ggx &\ggx &\ggx &\ggx &\ggx &     &     &     &     &     \\
		v_3 &    &     &     &\ggx &     &     &     &     &     &\ggx &\ggx &\ggx &\ggx &\ggx \\
	\end{array}$
	\caption{An election for removing an approval, $k=2$.}
	\label{fig-SAV-remove-RL}
	}
      \end{subfigure}
      \caption{Approval election matrices illustrating the proof of \Cref{pro:sav-add}.}
\end{figure}

The case of SAV is more intricate as the rule is $k$-{\ggadd}-robust
and $k$-{\ggremove}-robust, but $1$-{\ggswap}-robust. Intuitively,
adding or removing a single approval can affect the scores of many
candidates so, in consequence, it can lead to replacing the whole
committee. On the other hand, swapping an approval affects the scores
of at most two candidates.

\begin{proposition}\label{pro:sav-add}
$\sav$ is $k$-{\ggadd}-robust, $k$-{\ggremove}-robust, and $1$-{\ggswap}-robust.
\end{proposition}

\begin{proof}

  We start with the case of adding an approval. Let us fix some
  committee size $k$ and let $A = \{a_1, \ldots, a_k\}$ and
  $B = \{b_1, \ldots, b_k\}$ be two disjoint sets of candidates. We
  form an election with candidate set $C = A \cup B$ and two voters,
  $v_1$ and $v_2$, such that $v_1$ approves all the candidates from
  $A$ and $v_2$ approves all the candidates from $B$. This election is
  illustrated at \Cref{fig-SAV-add-RL}.  Each candidate has SAV score
  $\nicefrac{1}{k}$ and, in particular, $A$ is one of the winning
  committees. Yet, if we add an approval for $b_1$ to the vote of
  $v_1$, then the scores of all the candidates in $A$ decrease to
  $\nicefrac{1}{k+1}$, whereas the scores of the candidates in $B$
  remain unchanged (or increase, as in the case of $b_1$). In
  consequence, $B$ becomes the unique winning SAV committee, which
  witnesses $k$-{\ggadd}-robustness of $\sav$.

  We proceed with the \ggremove{} case.  We form an election
  $E=(\mathcal{C},V)$ with the committee size~$k$ and candidate set
  $\mathcal{C} = \{s\} \cup A \cup B \cup C \cup D$, where
  $ A=\{a_1,...,a_{k}\}, B=\{b_1,...,b_{k-1}\}, C=\{c_1,...,c_{k+3}\},
  D=\{d_1,...,d_{k+3}\}$.  The voter collection is
  $V=(v_1, v_2, v_3)$, where $v_1$ approves only candidates in
  $A\cup\{s\}$, $v_2$ approves candidates in $B \cup C$, and $v_3$
  approves candidates in $B \cup D$. The election is illustrated in
  \Cref{fig-SAV-remove-RL}.
  The winning committee score for election $E$ is $\frac{k}{k+1}$
  because the greatest score a candidate may get is $\frac{1}{k+1}$
  and there are at least $k$ such candidates coming from
  $A \cup B \cup \{s\}$.  To see this, note the following.  The
  candidates in $A \cup \{s\}$ have $\frac{1}{k+1}$ score each, coming
  from $v_1$.  The candidates in $B$ have
  $\frac{2}{k-1+k+3}=\frac{1}{k+1}$ score each, coming from voters
  $v_2$ and $v_3$.  Finally, the candidates in $C$ and $D$ have a
  lower score of $\frac{1}{k-1+k+3}=\frac{1}{2(k+1)}$.

  Let $W$ be a winning committee $B \cup \{s\}$.  We obtain an
  election $E'$ by removing $v_1$'s approval for~$s$.  Note that we
  now have $k$ candidates in $A$ with increased score, equal
  $\frac{1}{k}$, no other candidate's score changes ($v_1$ does not
  approve any other candidates) and their scores remain at most
  $\frac{1}{k+1}$.  It follows that the only winning committee of $E'$
  is $A$, with score $\frac{k}{k}=1$.  This means that no candidate in
  previously winning committee $W$ belongs to any of the $E'$'s
  winning committees. Hence, we get $k$-{\ggremove}-robustness of
  $\sav$.

  Finally, the case of swapping an approval is very different from
  adding or removing approvals for the $\sav$ rule.  It follows from
  the fact that a \ggswap{} can only happen within a single vote.
  Thus, unlike \ggadd{} and \ggremove{}, it doesn't change the actual
  score per single approval because the count of approval per voter
  doesn't change.  The proof of \ggrobustness{1}{SAV} follows exactly
  the same reasoning as the one for \ggrobustness{1}{AV}.
\end{proof}

For the case of unit-decreasing Thiele rules (including CC and PAV),
we find that they all are $k$-robust for each operation type.  Our
proof uses a single election (parametrized by the committee size) for
all the considered rules.  The result is expected as Bredereck et
al.~\cite{bre-fal-kac-nie-sko-tal:j:robustness} have shown that for
$\np$-hard rules (with polynomial-time computable scoring functions)
one cannot expect constant robustness values,\footnote{Formally, their
  result applies to the ordinal setting, but it is straightforward to
  adapt their reasoning to the approval setting.}  but it is appealing
to have a single, simple election which witnesses that all the rules
in the class are $k$-robust (and, indeed, the result of Bredereck et
al.~\cite{bre-fal-kac-nie-sko-tal:j:robustness} only says that the
robustness level is not constant, but does not suffice to claim that
it is $k$).

\begin{figure}
	{\footnotesize 
	\begin{subfigure}[b]{0.3\textwidth} 
	\centering
	\arraycolsep=1.0pt\def\arraystretch{1.0}
	$\begin{array}{c|ccccccc}
		  & a_1 & a_2 & a_3 & b_1 & b_2 & b_3 \\
		  \hline
		v^s     & +    &    &    &    &    &\\
		v_{1,1} & \ggx &    &    &\ggx&    &    \\
		v_{1,2} & \ggx &    &    &    &\ggx&    \\
		v_{1,3} & \ggx &    &    &    &    &\ggx\\
		v_{2,1} &      &\ggx&    &\ggx&    &    \\
		v_{2,2} &      &\ggx&    &    &\ggx&    \\
		v_{2,3} &      &\ggx&    &    &    &\ggx\\
		v_{3,1} &      &    &\ggx&\ggx&    &    \\
		v_{3,2} &      &    &\ggx&    &\ggx&    \\
		v_{3,3} &      &    &\ggx&    &    &\ggx\\
		\end{array}$
	\caption{\textsc{Add}}
	\label{fig-overlapping_add}
	\end{subfigure}
	\begin{subfigure}[b]{0.3\textwidth}
	\centering
	$\arraycolsep=1.0pt\def\arraystretch{1.0}
		\begin{array}{c|ccccccc}
		  & a_1 & a_2 & a_3 & b_1 & b_2 & b_3 \\
		  \hline
		v^s     & \ggx &    &    & -  &    &\\
		v_{1,1} & \ggx &    &    &\ggx&    &    \\
		v_{1,2} & \ggx &    &    &    &\ggx&    \\
		v_{1,3} & \ggx &    &    &    &    &\ggx\\
		v_{2,1} &      &\ggx&    &\ggx&    &    \\
		v_{2,2} &      &\ggx&    &    &\ggx&    \\
		v_{2,3} &      &\ggx&    &    &    &\ggx\\
		v_{3,1} &      &    &\ggx&\ggx&    &    \\
		v_{3,2} &      &    &\ggx&    &\ggx&    \\
		v_{3,3} &      &    &\ggx&    &    &\ggx\\
		\end{array}$
	\caption{\textsc{Remove}}
	\label{fig-overlapping_remove}
	\end{subfigure}
	\begin{subfigure}[b]{0.3\textwidth} 
	\centering
	\arraycolsep=1.0pt\def\arraystretch{1.0}
	$\begin{array}{c|ccccccc}
		  & a_1 & a_2 & a_3 & b_1 & b_2 & b_3 \\
		  \hline
		v^s     &  +   &    &    & -  &    &\\
		v_{1,1} & \ggx &    &    &\ggx&    &    \\
		v_{1,2} & \ggx &    &    &    &\ggx&    \\
		v_{1,3} & \ggx &    &    &    &    &\ggx\\
		v_{2,1} &      &\ggx&    &\ggx&    &    \\
		v_{2,2} &      &\ggx&    &    &\ggx&    \\
		v_{2,3} &      &\ggx&    &    &    &\ggx\\
		v_{3,1} &      &    &\ggx&\ggx&    &    \\
		v_{3,2} &      &    &\ggx&    &\ggx&    \\
		v_{3,3} &      &    &\ggx&    &    &\ggx\\
		\end{array}$
	\caption{\textsc{Swap}}
	\label{fig-overlapping_swap}
	\end{subfigure}
	\caption{Approval election matrices illustrating proofs of
          \Cref{prop-RL-UnitDecreasing} and
          \Cref{prop-RL-GreedyUnitDecreasing}, for $k=3$.}
	\label{fig-unit-decreasing-RL}
	}
\end{figure}

\begin{proposition}\label{prop-RL-UnitDecreasing}
  Every unit-decreasing Thiele rule is $k$-Op-robust for each
  $\mathrm{Op} \in \{\ggadd{},$ $\ggremove{},$ $\ggswap{}\}$.
\end{proposition}

\begin{proof}
  We start with the case of adding an approval.  Let us fix committee
  size $k$ and let $\ggR$ be some unit-decreasing Thiele rule which
  for the case of committees of size $k$ uses vector
  $\omega = (1, \omega_2, \ldots, \omega_k)$.  The main idea of the
  proof is similar to that used for \Cref{pro:sav-add}.  Let
  $A=\{a_1,...,a_k\}$ and $B=\{b_1,...,b_k\}$ be two disjoint sets of
  candidates. We form an election $E=(C,V)$ where $C=A \cup B$ and $V$
  contains the following voters (see \Cref{fig-overlapping_add} for
  illustration):
  \begin{enumerate}
  \item For each $i, j \in [k]$, there is a voter $v_{i,j}$ who
    approves candidates $a_i$ and $b_j$.
  \item There is a single voter $v^s$ with an empty approval set.
  \end{enumerate}
  Let $S$ be some size-$k$ committee. We note that
  $\score_E^{\omega\hbox{-}\av}(S) \leq k^2$. This is so, because
  under a unit-decreasing Thiele rule each approval for a given
  candidate $c$ can contribute at most one point toward the score of a
  committee that contains $c$. In our election each candidate is
  approved by exactly $k$ voters and the committee is of size $k$.  In
  particular, committee $B$ has score $k^2$ and is among the winning
  committees.

  We obtain election $E'$ by adding an approval for $a_1$ in the vote
  of $v^s$.  As a consequence, the highest possible score of a
  size-$k$ committee in $E'$ is $k^2+1$, and committee $A$ indeed
  obtains this score. We now show that $A$ is the unique winning
  committee in $E'$.
  First, we note that $a_1$ must belong to each winning committee as
  committees without $a_1$ obtain at most $k^2$ points.  Second, let
  us assume for the sake of contradiction that some committee $D$
  contains both $a_1$ and some candidate $b_\ell \in B$.  By
  construction, voter $v_{1,\ell}$ approves both $a_1$ and $b_\ell$
  and---by definition of a unit-decreasing Thiele rule---contributes
  less than $2$ points to the score of $D$. All the other voters
  provide exactly $k^2-1$ approvals for the members of $D$ and even if
  each of these approvals contribute a single point to the score of
  $D$, in total this score would be lower than $k^2+1$. Thus $A$ is
  the only winning committee in $E'$.
  As $B$ is among the winning committees in $E$ and is disjoint from
  $A$, we have that our rule is $k$-{\ggadd}-robust.
  
  The proof for $\ggR$ being \ggrobust{k}{Remove} is similar.  Let
  $E''$ be an election constructed identically as $E$ in the preceding
  case, except that voter $v^s$ approves both $a_1$ and $b_1$ (see
  \Cref{fig-overlapping_remove} for an illustration). First, we see
  that $B \in \ggR(E'',k)$ (with maximum score of $k^2+1$). Next, we obtain election $E'$ by
  removing $b_1$'s approval from $v^s$. Since we have that $B \in \ggR(E'',k)$, $\ggR(E',k)=\{A\}$, and
  $|B \cap A|=0$, we have that $\ggR$ is \ggrobust{k}{Remove}.

  Finally, to prove that $\ggR$ is \ggrobust{k}{Swap}, it suffices to
  note that $B$ is the only winning committee in an election obtained
  from $E'$ by swapping an approval from $a_1$ to $b_1$ in $v^s$ (see
  \Cref{fig-overlapping_swap}).
\end{proof}

The above results, as well as those of Bredereck et
al.~\cite{bre-fal-kac-nie-sko-tal:j:robustness}, give some intuitions
regarding robustness levels that we may expect from multiwinner rules.
On the one hand, simple, polynomial-time computable rules that focus
on individual excellence tend to have robustness levels equal to~$1$
(this includes, e.g., AV in the approval setting, and a number of
rules in the ordinal one; SAV is a notable exception).
Indeed,  Bredereck et
al.~\cite[Theorem~6]{bre-fal-kac-nie-sko-tal:j:robustness} have shown that  if a
rule selects a committee with the highest score, this score is easily
computable, and the rule's robustness level is bounded by a constant,
then some winning committee can be computed in polynomial time.
On the other hand, more involved rules that focus on proportionality
or diversity---in particular those $\np$-hard to compute---tend to
have robustness levels equal to the committee size.
However, regarding rules that form the committee sequentially, so far
there was only one data point: Bredereck et
al.~\cite{bre-fal-kac-nie-sko-tal:j:robustness} have shown that single
transferable vote (STV; a well-known rule for the ordinal setting,
focused on proportionality) has robustness levels equal to the
committee size. We provide further such examples by showing that
greedy variants of unit-decreasing Thiele rules, as well as
Phragm{\'e}n, also have robustness levels equal to the committee
size. We first show that robustness-levels based on adding and
removing approvals are equal for resolute rules (i.e., rules that
always output a unique winning committee, as is the case for our
sequential ones).

\begin{proposition}\label{prop:add-rem-robustness}
  Let $\calR$ be a resolute multiwinner voting rule, and let $\ell$ be
  a positive integer.  $\calR$ is $\ell$-\textsc{Add}-robust if and
  only if it is $\ell$-\textsc{Remove}-robust.
\end{proposition}
\begin{proof}
  Let us fix committee size $k$ and a resolute multiwinner rule
  $\calR$. Further, assume that $\calR$ is $\ell$-\textsc{Add}-robust
  for some integer $\ell$. We will show that it also is
  $\ell$-\textsc{Remove}-robust. To see this, let us consider two
  elections, $E$ and $E'$, where $E'$ is obtained from $E$ by removing
  an approval and both elections contain at least $2k$
  candidates. Let~$W$ be the unique winning committee in $\calR(E,k)$
  and $W'$ be the unique winning committee in $\calR(E',k)$. Since
  $\calR$ is $\ell$-\textsc{Add}-robust, we know that $W$ and $W'$
  differ by at most $\ell$ candidates (it suffices to apply the
  defintion of $\ell$-\textsc{Add}-robustness, but with the roles of
  $E$ and $E'$ reversed). Similarly, we see that there are two such
  elections whose winning committees differ by exactly $\ell$
  candidates.
  By applying analogous reasoning, we see that if $\calR$ is
  $\ell$-\textsc{Remove}-robust then it is also
  $\ell$-\textsc{Add}-robust.
\end{proof}

To establish robustness levels of our sequential rules, we use the
same constructions as for full-fledged Thiele rules.

\begin{proposition}\label{prop-RL-GreedyUnitDecreasing}
  Greedy unit-decreasing Thiele rules are \ggrobust{k}{Op} for all
  \ggop{} in $\{\ggadd{}$, $\ggremove{}$, $\ggswap{}\}$.
  The same holds for Phragm\'en.
\end{proposition}

\begin{proof}
  Let $\calR$ be a greedy unit-decreasing Thiele rule and consider
  elections $E$ and $E'$ as in the proof of
  \Cref{prop-RL-UnitDecreasing}, with committee size $k$ and
  tie-breaking order
  $b_1 \pref \cdots \pref b_k \pref a_1 \pref \cdots \pref a_k$
  (recall \Cref{fig-overlapping_add}). One can verify that
  $\calR(E,k) = B$ and $\calR(E',k) = A$ (for the former, we make use
  of the tie-breaking order to note that $b_1$ is selected in the
  first iteration, followed by $b_2$, $b_3$, and so on; for the
  latter, $a_1$ is selected in the first iteration because it has
  highest score, and then it is followed by $a_2$, $a_3$, and so
  on). This shows \ggrobustness{k}{\ggadd} and
  \ggrobustness{k}{\ggremove} of~$\calR$ (for the latter, we invoke
  \Cref{prop:add-rem-robustness}). To see \ggrobustness{k}{\ggswap},
  note that if in $E$ we extend $v^s$ with an approval for $b_1$ then
  $B$ still wins, but this election is one approval-swap away from
  $E'$, where $A$ wins.

  The same construction works for Phragm{\'e}n. Indeed, we see that
  $B$ is winning in election $E$: Each candidate is approved by $k$
  voters so, at time $\frac{1}{k}$, first $b_1$ is purchased (due to
  the tie-breaking order) and then $b_2$, $b_3$, and so on, up to
  $b_k$ (due to the tie-breaking order and because each of these
  candidates is approved by a disjoint group of $k$ voters).
  Similarly, $A$ is winning in $E'$: First $a_1$ is selected at time
  $\frac{1}{k+1}$ (as $a_1$ has $k+1$ approvals), and then
  $a_2, \ldots, a_k$ are selected at time $\frac{1}{k}$ (because each
  of them is approved by a disjoint group of $k$ voters and no
  candidate in $B$ is approved by voters who have a unit of money in
  total at this time; indeed, for each $b_i \in B$, at time
  $\frac{1}{k+1}$ one of the voters supporting $b_i$ spent his or her
  money on $a_1$).  This shows \ggrobustness{k}{\ggadd} and
  \ggrobustness{k}{\ggremove}. We obtain \ggrobustness{k}{\ggswap} in
  the same way as for the greedy unit-decreasing Thiele rules.
\end{proof}

\section{Complexity of the Robustness Radius Problems}

In this section we focus on the complexity of the
\ggR-\textsc{Op-Robustness-Radius} problems for adding, removing, and
swapping approvals. For the rules where our decision problems are
polynomial-time solvable, we also consider the complexity of the
respective counting problems.

\subsection{The AV Rule: Polynomial-Time Algorithms}
We start by considering the AV rule. In this case we have
polynomial-time algorithms for all our decision and counting problems
(recall that we do not consider the \textsc{Swap} operation for the
counting problems).

\begin{theorem}\label{propRR-av}
  \ggrr{\ggav{}-\ggop{}} is in P for each
  $\ggop{} \in \{\ggadd{}, \ggremove{}, \ggswap\}$.
\end{theorem}

\begin{proof}
  We first describe a polynomial-time algorithm for deciding
  \ggrr{\ggav{}-\ggadd{}}.  Consider an election $E = (C,V)$, let
  $m = |C|$ be the number of candidates, let $n = |V|$ be the number
  of voters, and let $k$ be the committee size.  We have
  $C = \{c_1, \ldots c_m\}$ and we assume that
  $\approval_E(c_1) \geq \approval_E(c_2) \geq \cdots \geq
  \approval_E(c_m)$. We can perform $B$ operaitons. There are three
  cases:
  \begin{enumerate}
  \item If $\approval_E(c_{k})>\approval_E(c_{k+1})$, then the first
    $k$ candidates form the only winning committee.  In order to
    change election result with the smallest number of \ggadd{}
    operations, we need to add approvals to $c_{k+1}$ until its score
    equals the score of $c_k$. Thus, we accept if
    $B \geq \approval(c_k)-\approval(c_{k+1})$, and otherwise we
    reject.

  \item If $\approval_E(c_{k})=\approval_E(c_{k+1}) < n$, then there
    exists some winning committees to which $c_k$ does not belong (and
    $c_{k+1}$ belongs).  It is enough to add one approval to $c_{k}$
    to ensure that $c_k$ is part of all the winning committees, thus
    changing the election result.  Therefore, we can accept if
    $B \geq 1$. Otherwise, we reject.

  \item If $\approval_E(c_{k})=\approval_E(c_{k+1}) = n$, then
    consider the smallest $t > k$ such that
    $n = \approval_E(c_{t})>\approval_E(c_{t+1})$. If such a $t$
    exists then we proceed as in the first case from this enumerated
    list, with $t$ taking the role of $k$. If such a $t$ does not
    exist, then it means that all the voters approve all the
    candidates and it is impossible to change the election result by
    adding approvals.
  \end{enumerate}

  A polynomial-time algorithm for \ggrr{\ggav{}-\ggremove{}} can be
  constructed analogously.  The differences are that in the first and
  second cases we remove approvals from $c_k$, in the second case
  instead of assuming $\approval_E(c_{k+1}) < n$ we assume
  $\approval_E(c_{k+1}) > 0$, and the third case occurrs when
  $\approval_E(c_{k}) = \approval_E(c_{k+1}) = 0$ (in this case we can
  change the result by removing all approvals from the candidate among
  $c_1, \ldots, c_k$ that has the smallest non-zero number of them; if
  all of them have zero approvals then it is impossible to change the
  result).

  For \ggrr{\ggav{}-\ggswap{}} we also proceed similarly.  If
  $\approval_E(c_k) > \approval_E(c_{k+1})$ then there are voters who
  approve $c_k$ but not $c_{k+1}$; we use them to move approvals from
  $c_k$ to $c_{k+1}$ until
  $\approval_E(c_{k+1}) \geq \approval_E(c_{k})$ and we accept if this
  happens within $B$ swap operations; we reject otherwise. If
  $\approval_E(c_k) = \approval_E(c_{k+1})$ then to change the result
  it suffices to move a single approval to or from either $c_k$ or
  $c_{k+1}$ (or between them). The only situation where this is
  impossible is if each voter either approves all candidate or none of
  them, but in this case no swap is possible at all.
\end{proof}
\begin{theorem}
  The problem of counting the number of ways in which $B$ approvals
  can be added (removed) without changing the set of \ggav{} winning
  committees is in FP.
\end{theorem}

\begin{proof}
  We focus on the case of adding approvals (the case of removing
  approvals is similar).
  Let $E=(C,V)$ be an approval-based election, where
  $C = \{c_1, \ldots, c_m\}$ and $V = (v_1, \ldots, v_n)$, and let $k$
  be the committee size.  For each candidate $c_i$, let
  $z_i = \approval_E(c_i)$ be its approval score. Without loss of
  generality, we assume that $z_1 \geq z_2 \geq \cdots \geq z_m$.  Our
  goal is to count the number of ways in which it is possible to add
  $B$ approvals to the election so that the result does not
  change.\footnote{We assume that it is, indeed, possible to add $B$
    approvals, i.e., $z_1 + \cdots + z_m \leq nm-B$.} We consider two
  cases, either there is a unique winning committee in $E$ or there
  are several winning committees.\bigskip

  \noindent\emph{Single Winning Committee.}\quad
  There is a unique winning committee, $W = \{c_1, \ldots, c_k\}$, in
  election $E$ exactly if $z_{k} > z_{k+1}$. We need to count the
  number of ways of adding $B$ approvals so that afterwards each member
  of $W$ has higher approval score than each candidate outside of
  $W$. To this end, for each $\ell, u \in [n]$ we define the following
  two values:
  \begin{enumerate}
  \item $f(\ell,u)$ is the number of ways of adding $B$ approvals so that
    each member of $W$ has approval score at least $\ell$ and each
    candidate outside of $W$ has approval score at most $u$.
      
  \item $g(\ell,u)$ is defined analogously to $f(\ell,u)$, except that we
    also require that at least one member of $W$ has approval score equal to $\ell$.
  \end{enumerate}

  Our algorithm should output the value
  $\sum_{\ell=1}^n g(\ell,\ell-1)$ as it covers all
  possibilities of adding $B$ approvals so that each member of $W$ has
  higher score than each candidate outside of $W$ without
  double-counting.
  To compute the values $g(\ell,u)$ in polynomial-time, we proceed as
  follows.  First, we note that for each $\ell,u \in[n]$, we have that
  $g(\ell,u) = f(\ell,u) - f(\ell+1,u)$.\footnote{For each
    $u \in [n]$, we take $f(n+1,u) =
    0$.} %
  This is so because subtracting $f(\ell+1,u)$ removes from
  consideration all the cases where all the members of $W$ have more
  than $\ell$ approvals.  Second, we note that values $f(\ell,u)$ can
  be computed using dynamic programming.
  \begin{lemma}\label{lem-f-polynomial}
    There is a polynomial-time algorithm that given $x,y \in [n]$
    computes the value $f(x,y)$ in polynomial-time.
  \end{lemma}
  \begin{proof}
    We assume that each candidate among $c_{k+1}, \ldots, c_m$ has at
    most $y$ approvals; otherwise we report that $f(x,y) = 0$.  For
    each candidate $c_i$, we compute values $a_i \leq b_i$ such that
    we need to add at least $a_i$ approvals to $c_i$, but no more than
    $b_i$. For each $c_i$ in the winning committee (i.e., if
    $1 \leq i \leq k$) we have $a_i = \max(0, x-\approval_E(c_i))$, to
    ensure that $c_i$ has at least $x$ approvals, and
    $b_i = n - \approval_E(c_i)$, because the only upper bound on the
    number of approvals we can add to $c_i$ stems from the number of
    the voters. For each $c_i$ among $c_{k+1}, \ldots, c_m$, we have
    $a_i = 0$, because it is fine not to add approvals at all, and
    $b_i = y - \approval_E(c_i)$ because $c_j$ cannot end up with more
    than $y$ approvals. For each $j \in [m]$ and
    $b \in [B] \cup \{0\}$, we define $h(j,b)$ to be the number of
    ways to add~$b$ approvals to candidates among $c_1, \ldots, c_j$
    so that each $c_i$ among them obtains between $a_i$ and $b_i$ new
    approvals. As a convention, we fix $h(j,b) = 0$ whenever $j$ is
    not in~$[m]$ or $[b]$ is not in~$[B] \cup \{0\}$, except that we
    set $h(0,0) = 1$. Next, we observe that for $j \in [m]$ and
    $b \in [B] \cup \{0\}$ we have:
    \[
      h(j,b) = \sum_{d=a_j}^{b_j} h(j-1,b-d) \cdot {n-\approval_E(c_j) \choose d}.
    \]
    Indeed, we need to add $d$ approvals to $c_j$, where $d$ is
    between $a_j$ and $b_j$, and we can do so in
    ${n-\approval_E(c_j) \choose d}$ ways. The remaining $b-d$
    approvals must be added to $c_1, \ldots, c_{j-1}$, which can be
    done in $h(j-1,b-d)$ ways. Using this formula and standard dynamic
    programming techniques, we obtain a polynomial-time algorithm for
    computing values $h$ and, hence, for computing $f(x,y)$.
  \end{proof}

  \noindent\emph{Several Winning Committees.}\quad
  There are several winning AV committees exactly when
  $z_k = z_{k+1}$.  Let $s$ and $t$ be two numbers such that
  $s \leq k$, $t \geq k+1$ and:
  \[
    z_1 \geq \cdots \geq z_{s-1} > z_s = \cdots = z_k = \cdots = z_t > z_{t+1} \geq \cdots z_m.
  \]
  In other words, candidates $c_s, c_{s+1}, \ldots c_t$ are the
  lowest-scoring members of the winning committees.

  To ensure that the set of winning committees does not change after
  adding $B$ approvals, we have to guarantee that:
  \begin{enumerate}
  \item Each of the candidates $c_1, \ldots, c_{s-1}$ has more approvals
    than each of the remaining candidates.
  \item Each of the candidates $c_s, c_{s+1}, \ldots, c_t$ has the same number of approvals.
  \item Each of the candidates $c_{t+1}, \ldots, c_m$ has fewer
    approvals than each of the remaining ones.
  \end{enumerate}
  This way we ensure that each winning committee consists exactly of
  the candidates $c_1, \ldots, c_{s-1}$ and $k-s+1$ arbitrarily chosen
  candidates from the set $\{c_s, \ldots, c_t\}$.  To count the number
  of ways of adding $B$ approvals so that our three conditions are
  fulfilled, we define function $g'(\ell,q,u)$ which gives the number
  of ways of adding $B$ approvals so that every candidate among
  $c_1, \ldots, c_{s-1}$ has at least $\ell$ approvals (and at least
  one of them has exactly $\ell$ of them), every candidate among
  $c_{s}, \ldots, c_{t}$ has exactly $q$ approvals, and every
  candidate among $c_{t+1}, \ldots, c_m$ has at most $u$
  approvals. Then the solution to our problem is
  $\sum_{\ell=1}^n \sum_{q < \ell} g'(\ell,q,q-1)$.  One can verify
  that it is possible to compute values of $g'$ using the same
  approach as in the case of a unique winning committee (indeed, one
  can compute $g'$ using $g$ from the previous case, by considering
  the candidates among $c_s, \ldots, c_t$ separately).

  For the case of removing approvals, we proceed identically, except
  that we change the definitions of functions $f$, $g$, and $g'$ to
  consider deleting approvals. Their computation, in particular the
  proof of \Cref{lem-f-polynomial} and the use of values $a_i$ and
  $b_i$, are analogous (modulo modifications due to removing approvals
  rather than adding them).
\end{proof}

\subsection{The SAV Rule: Easy Decision Problems,  Hard Counting Ones}

Let us now move on to the case of the SAV rule. The decision variants
of our three problems are still polynomial-time computable for this
rule.

\begin{theorem}\label{propRR-sav}
  \ggrr{\ggsav{}-\ggop{}} is in $\p$ for each
  $\ggop{} \in \{\ggadd{}, \ggremove{}, \ggswap\}$.
\end{theorem}

\begin{proof}

  We start by analyzing the case where there is a single winning
  committee (we will address the multiple winning committees case
  later).  Let $E=(C,V)$ be the input election, $k$ be the committee
  size. For each candidate $c \in C$, we let
  $\score(c) = \sum_{v \in V(c)}\frac{1}{|A(v)|}$ by $c$'s individual
  score, and we let $X \subseteq C$ be the $k$ candidates with the top
  scores. We let $Y = C \setminus X$.  This implies that the scores of
  candidates in $X$ are strictly higher than the scores of candidates
  in $Y$. We will write $E'$ to refer to hypothetical election after
  modifications.  To change the election result, there must exist an
  $x \in X$ and a $y \in Y$ such that in $E'$
  $\score(y) \geq \score(x)$.  Otherwise, the winning set $X$ would not
  change.  For each \ggop{} in \ggadd{}, \ggremove{}, \ggswap{} we
  construct a polynomial-time algorithm deciding
  \ggrr{\textsc{SAV}-\ggop{}}.  The proof will follow by constructing
  a polynomial-time algorithm that decides whether for a given
  $x \in X, y \in Y$ we can ensure $\score(y) \geq \score(x)$ by
  applying at most $B$ approval changes.  Having such an algorithm, we
  construct our target algorithm by simply iterating through all the
  polynomially many choices of $x,y$ and accepting if any of them
  leads to acceptance.  Otherwise, we reject.  In each case (\ggadd{},
  \ggremove{}, \ggswap{}), we start by computing $X$ in polynomial
  time.

  In the \ggadd{} case, we can ignore adding approvals to voters who
  already approve $y$ because even if it might decrease $\score(x)$ it
  would decrease $\score(y)$ by the same amount.  We are left with the
  voters that do not approve $y$.  But for those voters, the only
  candidate worth adding approvals to is $y$ rather then some other
  candidate $z$.  Indeed, after adding an approval for candidate $y$
  the score of $y$ increases and the score of $x$ decreases (if $x$
  were also approved).  Adding an approval for candidate $z$ decreases
  the score of $x$ by the same amount, but does not increase the score
  of $y$.  Once some voter $v$ approves $y$, we again note that there
  is no point in adding further approvals to $v$.  Furthermore, we
  note that adding an approval to a voter does not impact the scores
  distributed by any other voters. Consequently, given $x$ and $y$ we
  proceed as follows: First, for each voter (who does not approve $y$),
  we compute how much would $y$'s score increase relative to the score
  of $x$ if this voter approved $y$. We sort the voters in the
  nonincreasing order with respect to these values and add approvals
  to the first $B$ of them (or all of them, if there are fewer than
  $B$ voters who do not approve $y$). If in the end $y$ has at least
  as high a score as $x$, then we accept. Otherwise, we reject.

  For the \ggrr{\textsc{SAV}-\ggswap{}} algorithm, we follow a similar
  approach. For each vote $v$, it suffices to consider a single swap
  operation:
  \begin{enumerate}
  \item If $v$ approves $x$ but not $y$, then this operation is to move an approval from
    $x$ to $y$.
  \item If $v$ approves $x$ and $y$, then this operation is to move an approval from $x$
    to some other caniddate (if the voter approves all the candidates, then no operation
    is possible).
  \item If $v$ approves neither $x$ nor $y$, then this operation is to
    move an approval from some candidate other than $x$ to $y$ (if the
    voter does not approve any of the candidates, then no operation is
    possible).
  \item If $v$ does not approves $x$ but approves $y$, then there is no operation that
    we may wish to perform on this vote.
  \end{enumerate}
  For each voter, we compute how much the score of $y$ would increase
  relative to that of $x$ if we performed the associated
  operation. Then we perform up to $B$ operations that give greatest
  gain to $y$ relative to $x$. If this leads to $y$ having higher or
  equal score as $y$ then we accept. Otherwise, we reject.

  \ggremove{} is more involved because it may be beneficial to remove
  more than a single approval from the voters. We will show, however,
  that a case-by-case analysis of voter groups gives a polynomial-time
  greedy algorithm.  Let $\Delta(x,y)=\score(x)-\score(y)$ ($\Delta$
  in short) and note that we aim to decrease $\Delta $ to zero or
  less.  First, observe that (i) we can ignore voters who do not
  approve neither $x$ nor $y$ because removing approvals from such
  voters does not impact $\Delta$.  Second, (ii) in the group of
  voters approving $y$ but not $x$, removing approvals (not for $y$)
  increases $\score(y)$ and thus decreases $\Delta$ and we should
  start with removing approvals from voters with lowest number of
  approvals (because the impact on $\Delta$ is highest then, as the
  score for a voter is inversely proportional to the count of voters).
  Third, (iii) in the group of voters approving $x$ but not $y$,
  approvals should be removed only from $x$ and the removals should
  start from voters with lowest number of approvals (as per argument
  given for the previous group).  Finally, (iv) in the group of voters
  approving both $x$ and $y$, we must start with removing approval for
  $x$ from voters with lowest count of approvals and once we remove an
  approval for $x$, the voter becomes part of group (ii).  Since the
  only dependencies between the groups of voters are that, after
  removing approval from $x$, group (iv) voters can become (ii) voters
  and group (iii) voters can become (i) voters, our algorithm can
  guess the number of removals to be applied to each of the groups (at
  most $B$ in total) and then apply them according to deterministic
  rules defined above, starting with group (iv) and (iii).  There are
  at most $O(B^2$) of such guesses (we do not remove from group (i))
  and our polynomial-time algorithm can simply iterate over all of the
  possibilities.  We accept if there exists an iteration where within
  $B$ removals we get $\Delta \leq 0$. Otherwise, we reject.

  What is left, is to address the situation where there are multiple
  winning committees. In this case there are at least two candidates
  that have the same score and whose score is lowest among those that
  belong to some winning committee(s). Try all pairs of such
  candidates. If it is possible to add an approval to one of them, or
  remove an approval from one of them, or move an approval from one of
  them (possibly to the other one), then with a single operation we
  change the result of the election (because the two candidates we
  consider end up with different scores). Now let us consider what
  happens if performing these operations is impossible:
  \begin{description}
  \item[Case 1:] For every pair of candidates we consider, it is
    impossible to add an approval to either of them. This means that
    each of the candidates that has lowest score in some winning
    committee is approved by all the voters. However, this implies
    that all the candidates are approved by all the voters and, hence,
    it is impossible to change the result by adding approvals.

  \item[Case 2:] For every pair of candidates we consider, it is
    impossible to remove approvals from either of them. This means
    that each of the candidates that has lowest score in some winning
    committee is not approved by any of the voters. Then to change the
    election result it suffices to find a candidate who is approved by
    fewest voters (but a nonzero number of them).  If this candidate
    is approved by at most $B$ voters then we can remove these
    approvals and change the election result. Otherwise, changing the
    result is impossible. Similarly, it is impossible to change the
    election result if neither of the candidates is approved by any of
    the voters.

  \item[Case 3:] For every pair of candidates we consider, it is impossible to
    move approvals. This means that each voter either approves all the
    candidates or none of them. Hence it is impossible to change the
    election result by moving approvals (because approval moves are
    impossible).    
  \end{description}
  This completes the proof.  
\end{proof}

In contrast to the case of the AV rule, for SAV the problems of
counting the number of ways in which $B$ approvals can be added or
removed without changing the election result are $\sharpp$-hard. We
show this by giving a Turing reduction from the classic
$\sharpp$-complete problem,
\textsc{\#Perfect-Matchings}~\cite{val:j:permanent}.\footnote{A
  (polynomial-time) Turing reduction of a counting problem $\#A$ to a
  counting problem $\#B$ is a polynomial-time procedure that given an
  instance of $\#A$ computes its result using $\#B$ as an oracle. This
  means that the procedure can form instances of problem $\#B$ and
  receives results for them in constant time (from the oracle). If
  $\#A$ Turing-reduces to $\#B$, this means that $\#A$ is no harder
  than $\#B$ because the ability to solve $\#B$ in polynomial-time
  yields the same ability for $\#A$.}  In this problem we are
given some bipartite graph $G$ with set $V(G)$ of the vertices ``on
the right'' and set $U(G)$ of the vertices ``on the left,'' where we
have that $|V(G)| = |U(G)|$. Each edge connects some vertex from
$V(G)$ with some vertex from $U(G)$. A perfect matching is a set of
edges such that every vertex belongs to exactly one edge in this
set. In \textsc{\#Perfect-Matchings} we ask for the number of perfect
matchings in our graph.

\begin{theorem}\label{thm:sav-add-counting}
  The problem of counting the number of ways in which approvals can be
  added without changing the set of SAV winning committees is
  $\sharpp$-complete.
\end{theorem}

\begin{proof}
  To see that the problem belongs to $\sharpp$, it suffices to note
  that a nondeterministic Turing machine can guess which approvals to
  add and then accept if this leads to changing the election result
  (which can be verified in deterministic polynomial time).
  
  To show $\sharpp$-hardness, we give a reduction from
  \textsc{\#Perfect-Matchings}, where we are given a bipartite graph
  $G$ and we ask how many perfect matchings it has. We write $V(G)$
  and $U(G)$ to denote the two sets of vertices in the graph, and we
  write $E(G)$ to denote its set of edges (so each edge from $E(G)$
  connects vertex from $V(G)$ with some vertex from $U(G)$).

  Let $n = |V(G)| = |U(G)|$ be the number of vertices in each part of
  the input graph (if $V(G)$ and $U(G)$ were of different
  cardinality, then there would be no perfect matching in our
  graph). We form an election $E = (C,V)$ as follows. We let the
  candidate set be $C = U(G) \cup V(G) \cup D$, where $D$ is a set of
  dummy candidates. Whenever we form a new voter who approves some
  dummy candidates, these dummy candidates are unique for this voter
  and are not approved by the other ones (thus, whenever we speak of a
  voter approving some dummy candidates, we implicitly create new
  dummy candidates). We create two groups of voters:
  \begin{enumerate}
  \item For each edge $e = \{u,v\} \in E(G)$, we form a voter $v_e$
    who approves vertex candidates $u$ and $v$, and $n-2$ dummy
    candidates (thus, this voter contributes $\nicefrac{1}{n}$ points to
    the SAV scores of $u$ and $v$). We refer to the voters in this
    group as to the \emph{edge voters}.

  \item For each vertex candidate $x \in U(G) \cup V(G)$, we form
    sufficiently many voters who each approve $x$ and $n-1$ dummy
    candidates, so that the SAV score of $x$ is exactly $2$ (we add at
    most $2n$ voters for each vertex candidate; note that the edge
    voters contribute at most $1$ point to the SAV score of each
    vertex candidate). We refer to the voters in this group as to the
    \emph{filler voters}.
  \end{enumerate}

  We set the size of the committee to be $k = 1$. In our election,
  each vertex candidate has SAV score $2$ and each dummy candidate has
  SAV score $\nicefrac{1}{n}$. Thus the winning committees are exactly
  the singleton subsets of $V(G) \cup U(G)$.

  We will show that the number of ways in which it is possible to add
  $n$ approvals to our election without changing the set of winning
  committees is of the form $f(G)\cdot M(G)$, where $M(G)$ is the
  number of perfect matchings of $G$ and $f(G)$ is an easily
  computable function. This suffices to show $\sharpp$-hardness of our
  problem. The intuition is that we can add $n$ approvals without
  changing the set of winning committees if and only if these
  approvals are added to the edge voters that correspond to a perfect
  matching (this way we decrease the scores of all the vertex
  candidates equally).  Let us now consider what happens after we add
  some $n$ approvals:

  \begin{description}
  \item[Some vertex candidate obtains an additional approval.] Let us
    assume that some vertex candidate $x$ obtains an additional
    approval in some vote $v$ (either a filler vote or an edge
    vote). Further, let $i-1$ be the number of other candidates that
    also obtain an additional approval in this vote.  Thus, $x$'s
    score increases by $\frac{1}{n+i}$ due to the additional
    approval. By adding the remaining $n-i$ approvals, we can decrease
    $x$'s score by at most $\frac{n-i}{n(n+1)}$. This is so, because
    adding a single approval to a vote where $x$ is already approved
    decreases $x$'s score by at most
    $\frac{1}{n} - \frac{1}{n+1} = \frac{1}{n(n+1)}$ (this happens
    when we add a single approval to an edge or filler vote where $x$
    is approved; adding two approvals to such a vote decreases the
    score by $\frac{1}{n}-\frac{1}{n+2} = \frac{1}{n(n+2)}$, which is
    less than $2\cdot \frac{1}{n(n+1)}$; in general, adding some $j$
    approvals decreases $x$'s score by a lesser value than adding
    single approvals to $j$ votes where $x$ has originally been
    approved). However, we have that:
    \[
\textstyle      \frac{1}{n+i} - (n-i)\cdot \frac{1}{n(n+1)} =
      \frac{n+i^2}{n(n+1)(n+i)} > 0.
    \]
    This means that after adding approvals, $x$'s score is greater
    than $2$. Yet, by adding $n$ approvals it is impossible to
    increase the scores of more than $n$ of the $2n$ vertex
    candidates. This means that if some vertex candidate obtains an
    additional approval, then the set of winning committees changes.

  \item[Some dummy vertex obtains an additional approval in a filler
    vote.] If some dummy candidate obtains an additional approval in a
    filler vote, then the score of the single vertex candidate who is
    approved in this vote decreases.\footnote{From the previous case
      we know that it is not possible to increase the score of this vertex
      candidate to the original value by adding some up to $n-1$ approvals for him
      or her.} The remaining $n-1$ additional approvals can lead to
    decreasing the scores of at most $2(n-1)$ vertex candidates (e.g.,
    if we add approvals in the edge votes only) so there is at least
    one vertex candidate whose score does not decrease. In
    consequence, the set of winning committees changes.

  \item[Approvals are added only for dummy candidates in the edge
    votes.] Adding an approval for a dummy candidate in an edge vote
    $v_e$ decreases the scores of the vertex candidates incident to
    $e$ by exactly $\frac{1}{n(n+1)}$.  Thus the only possibility that
    we add $n$ approvals and the scores of all the vertex candidates
    remain the same (but decreased) is that we add approvals to dummy
    candidates in $n$ edge votes that form a perfect matching (simple
    counting arguments %
    show that if we add two approvals to some edge vote or if the edge
    votes do not form a perfect matching, then some vertex candidate's
    score does not decrease).
  \end{description}

  Let $M$ be the number of perfect matchings in our graph. By the
  above reasoning, for each perfect matching we have exactly
  $(|D|-(n-2))^n$ ways to add $n$ approvals in our election so that
  the set of winning committees does not change (for each edge from
  the matching, in the corresponding vote we can add an approval for
  one of the $|D|-(n-2)$ dummy candidates that are not approved
  there). Thus, given a solution for our problem, it is easy to obtain
  the number of perfect matchings in the input graph.
\end{proof}

\begin{theorem}\label{thm:sav-remove-counting}
The problem of counting the number of ways in which
approvals can be removed without changing the set of SAV winning
committees is $\sharpp$-complete.
\end{theorem}
The proof is quite similar to the one for \Cref{thm:sav-add-counting} and we present it in \Cref{app:sav-remove-counting}.

\subsection{Unit-Decreasing Thiele Rules: Hardness Results}

For the case of unit-decreasing Thiele rules, all our decision
problems are $\np$-hard (and, so, there is no point in considering
their counting variants). To show this fact, we give reductions from
the \textsc{X3C} problem. Naturally, our results apply to CC and
PAV.

\begin{definition}\label{def:x3c}
  An instance of \textsc{X3C} consists of a universe set $U = \{u_1,$
  $\ldots, u_{3k}\}$ and a family $\mathcal{S} = \{S_1, \ldots, S_m\}$ of
  tree-element subsets of $U$. We ask if there is a collection
  $\mathcal{S}^{cov} \subseteq \mathcal{S}$ of $k$ sets from $\mathcal{S}$
  whose union is $U$ (i.e., we ask if there is an exact cover of $U$).
\end{definition}

\begin{theorem}\label{propNPSwapUD}
  Let \ggR{} be a unit-decreasing Thiele rule and let $\ggop$ be an
  operation in $\{\ggadd, \ggremove, \ggswap\}$. \ggrr{\ggR{}-\ggop{}}
  is $\np$-hard.
\end{theorem}

\begin{proof}
  We start with the case of the $\ggswap$ operation and give a
  reduction from the classic $\np$-complete problem
  \textsc{Exact-Cover-By-3-Sets} (X3C). Our input X3C instance is as
  described above, in \Cref{def:x3c}. Namely, we have a universe set
  $U = \{u_1, \ldots, u_{3k}\}$, a family
  $\mathcal{S} = \{S_1, \ldots, S_m\}$ of size-$3$ subsets of $U$, and
  we ask for the existence of an exact cover of $U$ with $k$ sets from
  $\mathcal{S}$.
  
  Let us consider committees of size $k$ and let
  $\omega = (1, \alpha, \omega_3, \ldots, \omega_k)$ be the vector
  that specifies the unit-decreasing Thiele rule $\ggR$ for this
  setting. By definition, we have $1 > \alpha$.  Let $\ell$ be the
  constant $\lceil \frac{3}{1-\alpha} \rceil$.
  We form election $E=({C}, {V})$ as follows.
  Let $\gC = \{\gc_1,...,\gc_m\}$ and $\gCC = \{\gcc_1,...,\gcc_k\}$
  be two sets of candidates, where the candidates in $\gC$ correspond
  to the sets from $\mathcal{S}$, and the candidates from $\gCC$ form
  a default winning committee.  We let the candidate set of our
  election be $C = A \cup B$ and we have the following voter
  groups (see Figure~\ref{fig-unit}):
  \begin{enumerate}
  \item We introduce voters $v_1, \ldots, v_{3k}$ who correspond to
    the elements of the universe set~$\mathcal{E}$ (we refer to them
    as the element voters). Each voter $v_i$ approves candidate
    $\gcc_{\lceil i/3 \rceil}$ and exactly those candidates $\gc_j$
    for whom it holds that element $u_i$ belongs to the set $S_j$.

  \item We introduce voter set
    $\gVV = \{\gvv_{i,j,l} \mid i \in [m], j \in [k], l \in [\ell]\}$,
    whose role is to form a mutual exclusion gadget regarding the
    winning committees (see details below). For each
    $i \in [m], j \in [k], l \in [\ell]$, $\gvv_{i,j,l}$ approves
    $\gc_i$ and $\gcc_j$. 

  \item We introduce voter set
    $\gVV^a = \{\gvv^a_{i,j,l} \mid i \in [m], j \in [m-k], l \in
    [\ell]\}$, whose role is to balance the scores of candidates in
    \gC{} and \gCC{}. For each $i \in [m], j \in [m-k], l \in [\ell]$,
    $\gvv^a_{i,j,l}$ approves $\gc_i$.

  \item We introduce voter set $\gVVV = \{\gvvv_1, \gvvv_2\}$,  whose members control
    which committee wins. Both voters in \gVVV{} approve $\gcc_1$.

  \end{enumerate}

  We show that $\gCC$ is the unique winning committee in election $E$.
  To this end, let us consider some arbitrary size-$k$ committee $Q$.
  Members of $Q$ can receive at most $3k$ approvals from the element
  voters, $km\ell$ approvals from the voters in the sets
  $\gVV \cup \gVV^a$, and $2$ approvals from the voters in \gVVV{}.
  Thus, $\score_E^{\omega\hbox{-}\av}(Q) \leq 3k+km\ell+2$ (and for
  the equality to hold, it cannot be the case that some voter approves
  more than one committee member).  Since one can verify that
  $\score_E^{\omega\hbox{-}\av}(\gCC)=3k+km\ell+2$, \gCC{} is a
  winning committee.  Furthermore, $B$ is a unique winning
  committee. To see this, let $Y$ be some winning committee different
  from $B$. Note that $Y$ must include candidate $b_1$ as this is the
  only way to get points from the voters in \gVVV{}. For the sake of
  contradiction, suppose that $Y$ also contains some candidate
  $\gc_i \in \gC{}$.  However, this means that there are $\ell$ voters
  in $\gVV$ who approve both $\gc_i$ and $\gcc_1$. By definition of
  the unit-decreasing Thiele rules, the score that committee $Y$
  receives from the $2\ell$ approvals that these voters grant to
  $\gc_i$ ad $\gcc_1$ is at most $\ell + \alpha\ell < 2\ell$ and, so,
  the total score of committee $Y$ is lower than $3k+mk\ell+2$.  Thus,
  we have that $\ggR{}(E, k) = \{\gCC{}\}$.

\begin{figure}
\footnotesize 
\centering
\arraycolsep=1.0pt\def\arraystretch{1.0}
$\begin{array}{c|cccc|cc}
      & \gc_1 & \gc_2 & \gc_3 & \gc_4 & \gcc_1 & \gcc_2 \\
    \hline
\gv_{1}             &\ggx&    &    &    &\ggx&    \\
\gv_{2}             &\ggx&    &    &\ggx&\ggx&    \\
\gv_{3}             &\ggx&\ggx&    &\ggx&\ggx&    \\
\gv_{4}             &    &\ggx&\ggx&\ggx&    &\ggx\\
\gv_{5}             &    &\ggx&\ggx&    &    &\ggx\\
\gv_{6}             &    &    &\ggx&    &    &\ggx\\
\hline
\gvv_{1,1,[\ell]}   &\ggX&    &    &    &\ggX&    \\
\gvv_{1,2,[\ell]}   &\ggX&    &    &    &    &\ggX\\
\gvv_{2,1,[\ell]}   &    &\ggX&    &    &\ggX&    \\
\gvv_{2,2,[\ell]}   &    &\ggX&    &    &    &\ggX\\
\gvv_{3,1,[\ell]}   &    &    &\ggX&    &\ggX&    \\
\gvv_{3,2,[\ell]}   &    &    &\ggX&    &    &\ggX\\
\gvv_{4,1,[\ell]}   &    &    &    &\ggX&\ggX&    \\
\gvv_{4,2,[\ell]}   &    &    &    &\ggX&    &\ggX\\
\hline
\gvv^a_{1,1,[\ell]} &\ggX&    &    &    &    &    \\
\gvv^a_{1,2,[\ell]} &\ggX&    &    &    &    &    \\
\gvv^a_{2,1,[\ell]} &    &\ggX&    &    &    &    \\
\gvv^a_{2,2,[\ell]} &    &\ggX&    &    &    &    \\
\gvv^a_{3,1,[\ell]} &    &    &\ggX&    &    &    \\
\gvv^a_{3,2,[\ell]} &    &    &\ggX&    &    &    \\
\gvv^a_{4,1,[\ell]} &    &    &    &\ggX&    &    \\
\gvv^a_{4,2,[\ell]} &    &    &    &\ggX&    &    \\
\hline
\gvvv_1             &    &    &    &    &\ggx&    \\
\gvvv_2             &    &    &    &    &\ggx&    \\
\end{array}$
\caption{Example of an election used in the proof of
  Theorem~\ref{propNPSwapUD}, for X3C instance with universe set
  $U = \{u_1, \ldots, u_6\}$ and sets $S_1 = \{u_1,u_2,u_3\}$,
  $S_2 = \{u_3,u_4,u_5\}$, $S_3 = \{u_4,u_5,u_6\}$, and
  $S_4 = \{u_2,u_3,u_4\}$. Consequently, we have $m=4$ and $k = 2$.
  Symbol $\ggx$ represents an approval for a candidate from a given
  voter, and $\ggX$ represents $\ell$ approvals coming from a group of
  voters.  }

  \label{fig-unit}
\end{figure}

  We now show that it is possible to change the set of winning
  committees with a single approval swap if and only if the answer for
  our input X3C instance~is~\emph{yes}.

  For the first direction, suppose that an exact cover
  $\mathcal{S}^{cov}$ exists for our input X3C instance and let $X$ be
  a size-$k$ committee corresponding to $\mathcal{S}^{cov}$ (i.e., $X$
  contains members of $\gC$ that correspond to the sets in
  $\mathcal{S}^{cov}$).  We note that
  $\score_E^{\omega\hbox{-}\av}(X)=3k+km\ell$.  To see this, note that
  as $X$ corresponds to a set cover, it receives $3k$ points from the
  element voters, each of the $k$ candidates in $X$ is approved by
  $k\ell$ voters from \gVV{}, and is approved by $(m-k)\ell$ voters
  from $\gVV^a$; each of these approvals translates to a single point
  because the voters in these groups do not approve more than one
  member of $X$ each. Finally, there are no approvals for members of
  $X$ from the voters in \gVVV{}.

  We form election $E'$ by swapping $\gvvv_2$'s approval from $\gcc_1$
  to some member of $X$, so we have
  $\score_{E'}^{\omega\hbox{-}\av}(X) = 3k+km\ell+1$. One can also verify
  that $\score_{E'}^{\omega\hbox{-}\av}(\gCC) = 3k+km\ell+1$. As a
  consequence (and following similar reasoning as a few paragraphs
  above) we see that both $X$ and $\gCC$ are winning committees.  In
  other words, a single approval swap sufficed to change the set of
  winning committees.

  For the other direction, let us assume that it is possible to change
  the set of winning committees in election $E$ by a single approval
  swap. We will show that this implies that there is an exact set
  cover for our input instance of X3C.  We first note that after a
  single swap, the score of committee $B$ can drop to no less than
  $3k + km\ell + 1$. Now, consider some committee $Y$ that contains
  candidates from both $\gC$ and $\gCC$. By a reasoning analogous to
  that showing that $\gCC$ is the unique winning committee in election
  $E$, members of $Y$ can, altogether, receive at most $3k+km\ell + 3$
  approvals,\footnote{We have the $+3$ component instead of the $+2$
    component due to the approval swap.} but there are at least $\ell$
  voters that approve at least two members of $Y$, each. Thus, the
  final score of committee $Y$ is at most
  $3k+mk\ell+3-\ell+\ell\alpha =3k+mk\ell+3-\ell(1-\alpha) \leq
  3k+mk\ell+3-\frac{3(1-\alpha)}{1-\alpha} =3k+mk\ell$. As a
  consequence, $Y$ certainly is not a winning committee.

  So, if after the approval swap there is a winning committee $X'$
  other than $\gCC$, it must consist of $k$ members of $\gC$. Further,
  both its score and the score of $\gCC$ must be exactly
  $3k+mk\ell+1$. Indeed, this is the lowest score that $\gCC$ may
  have, and the highest score that $X'$ may have after the
  swap. Further, $X'$ would have such score only if it received
  exactly $3k+km\ell$ points prior to the swap.  However, $X'$ can
  have score $3k+km\ell$ prior to the swap only if it corresponds to
  an exact cover for our X3C instance (otherwise the element voters
  would assign fewer than $3k$ points to $X'$). This concludes the
  proof.
  
  The proofs of \ggremove{} and \ggadd{} cases follow a similar
  approach.  The difference is that we do not need the voter
  $\gvvv{}_2$ and we either \ggremove{} voter $\gvvv_1$'s approval
  from $\gcc{}_1$ or \ggadd{} $\gvvv_1$'s approval to some
  $c \in \gC{}$.
\end{proof}

One could seek various workarounds for the above hardness
result. Indeed, in one of the conference papers on which this one is
based, we reported existence of $\fpt$ algorithms parameterized by
either the number of candidates or the number of
voters~\cite{gaw-fal:c:approval-robustness}. We decided to omit these
resutls here to streamline the presentation (for readers who wish to
recreate them, their proofs were based on ILP formulations and
techniques used in the work of \citet{fal-sko-tal:c:bribery-success}).

\subsection{GreedyCC, GreedyPAV, and Phragm{\'e}n: NP-Completeness}

In this section we show that the \textsc{Robustness-Radius} problem is
$\np$-complete for GreedyCC and GreedyPAV, for adding, removing, and
swapping approvals. We observe that for each of our rules and
operation type, the respective \textsc{Robustness-Radius} problem is
clearly in $\np$. Indeed, it suffices to nondeterministically guess
which approvals to add/remove/swap, compute the winning committees
before and after the change (since our rules are resolute, in each
case there is exactly one), and verify that they are different. Hence,
in our proofs we will focus on showing $\np$-hardness.  We give
reductions from a restricted variant of \textsc{X3C}, which we call
\textsc{RX3C} (the difference, as compared to \textsc{X3C} in
\Cref{def:x3c}, is that every universe element belongs to exactly
three sets; it is well known that this variant of the problem remains
$\np$-complete~\cite{gon:j:x3c}).

\begin{definition}\label{def:rx3c}
  An instance of \textsc{RX3C} consists of a universe set $U = \{u_1,$
  $\ldots, u_{3n}\}$ and a family $\mathcal{S} = \{S_1, \ldots, S_{3n}\}$ of
  three-element subsets of $U$, such that each member of $U$ belongs
  to exactly three sets from~$\mathcal{S}$. We ask if there is a collection
  of $n$ sets from $\mathcal{S}$ whose union is $U$ (i.e., we ask if there
  is an exact cover of $U$).
\end{definition}

All our reductions follow the same general scheme: Given an instance
of \textsc{RX3C} we form an election where the sets are the candidates
and the voters encode their content. Additionally, we also have
candidates $p$ and $d$. Irrespective which operations we perform
(within a given budget), all the set candidates are always selected,
but by performing appropriate actions we control the order in which
this happens. If the order corresponds to finding an exact cover, then
additionally candidate~$p$ is selected. Otherwise, our rules select
$d$.  We first focus on adding approvals and then argue why our proofs
adapt to the other cases.  We mention that the proofs of the next two
theorems were adapted by \citet{jan-fal:c:ties-approval-voting} to the
general case of greedy unit-decreasing Thiele rules, albeit for the
problem of detecting ties.\footnote{In our paper, we assume resolute
  variants of the sequential rules, so ties are
  impossible. \citet{jan-fal:c:ties-approval-voting} assume
  parallel-universes tie-breaking and ties can happen. On the
  technical level, their proof is very similar to ours and could
  easily be adapted to showing $\np$-completeness of the
  \textsc{Robustness-Radius} problems for greedy unit-decreasing
  Thiele rules, but---formally---they show something a bit different.}

\begin{theorem}\label{thm:greedy-cc}
  GreedyCC-\textsc{Add-Robustness-Radius} is $\np$-complete.
\end{theorem}
\begin{proof}
  We give a reduction form the \textsc{RX3C} problem. Our input
  consists of the universe set $U = \{u_1, \ldots, u_{3n}\}$ and
  family $\calS = \{S_1, \ldots, S_{3n}\}$ of three-element subsets of
  $U$. We know that each member of $U$ belongs to three sets
  from~$\calS$.
  We introduce two integers, $T = 10n^5$ and $t = 10n^3$, and we
  interpret both as large numbers, with $T$ being significantly larger
  than $t$.  We form an election $E = (C,V)$ with candidate set
  $C = \{S_1, \ldots, S_{3n}\} \cup \{p,d\}$, and with the following
  voters:
  \begin{enumerate}
  \item For each $S_i \in \calS$, there are $T$ voters that approve
    candidate $S_i$.
  \item For each two sets $S_i$ and $S_{j}$, there are $T$ voters
    that approve candidates $S_i$ and $S_{j}$.
  \item There are $2nT+4nt$ voters that approve $p$ and $d$.
  \item For each $u_\ell \in U$, there are $t$ voters that approve $d$
    and those candidates $S_i$ that correspond to the sets containing
    $u_\ell$.
  \item There are $n$ voters who do not approve any candidates.
  \end{enumerate}
  The committee size is $k = 3n+1$ and the budget, i.e., the number of
  approvals we can add, is $B = n$.  We assume that the tie-breaking
  order among the candidates is:
  \[
    S_1 \pref S_2 \pref \cdots \pref S_{3n} \pref p \pref d.
  \]
  Prior to adding approvals, each candidate $S_i$ is approved by
  $3nT+3t$ voters, $p$ is approved by $2nT+4nt$ voters, and $d$ is
  approved by $2nT+7nt$ voters.

  Let us now consider how GreedyCC operates on this election. Prior to
  the first iteration, all the set candidates have the same score,
  much higher than that of $p$ and $d$. Due to the tie-breaking order,
  GreedyCC chooses $S_1$. As a consequence, all the voters that
  approve $S_1$ are removed from consideration and the scores of all
  the other set candidates decrease by $T$ (or by $T+t$ or $T+2t$, for
  the sets that included the same one or two elements of $U$ as
  $S_1$). GreedyCC acts analogously for the first $n$ iterations,
  during which it chooses a family $\calT$ of $n$ set elements (we
  will occasionally refer to $\calT$ as if it really contained the
  sets from~$\calS$, and not the corresponding candidates).
  
  After the first $n$ iterations, each of the remaining $2n$ set
  candidates either has $2nT$, $2nT+t$, $2nT+2t$, or $2nT+3t$
  approvals (depending how many sets in $\calT$ have nonempty
  intersection with them). Let $x$ be the number of elements from $U$
  that do not belong to any set in $\calT$.  Candidate $p$ is still
  approved by $2nT+4nt$ voters, whereas $d$ is approved by
  $2nT+4nt + xt$ voters. Thus at this point there are two
  possibilities. Either $x = 0$ and, due to the priority order,
  GreedyCC selects $p$, or $x > 0$ and GreedyCC selects $d$. In either
  case, in the following $2n$ iterations it chooses the remaining $2n$
  set candidates (because after the $n+1$-st iteration the score of
  that among $p$ and $d$ who remained, drops to zero or nearly
  zero). If candidate $p$ is selected without adding any approvals,
  then it means that we can find a solution for the \textsc{RX3C}
  instance using a simple greedy algorithm. In this case, instead of
  outputting the just-described instance of
  GreedyCC-\textsc{Add-Robustness-Radius}, we output a fixed one, for
  which the answer is \emph{yes}. Otherwise, we know that without
  adding approvals the winning committee is
  $\{S_1, \ldots, S_{3n},d\}$. We focus on this latter case.
  
  We claim that it is possible to ensure that the winning committee
  changes by adding at most $n$ approvals if and only if there is an
  exact cover of $U$ with~$n$ sets from $\calS$.  Indeed, if such a
  cover exists, then it suffices to add a single approval for each of
  the corresponding sets in the last group of voters (those that
  originally do not approve anyone). Then, by the same analysis as in
  the preceding paragraph, we can verify that the sets forming the
  cover are selected in the first $n$ iterations, followed by $p$,
  followed by all the other set candidates.
  
  For the other direction, let us assume that after adding some $t$
  approvals the winning committee has changed. One can verify that
  irrespective of which (up to) $n$ approvals we add, in the first $n$
  iterations GreedyCC still chooses $n$ set candidates. Thus, at this
  point, the score of $p$ is at most $2nT+4nt+n$ and the score of $d$ is
  at least $2nT+4nt+xt-n$ (where $x$ is the number of elements from
  $U$ not covered by the chosen sets; we subtract $n$ to account for
  the fact that $n$ voters that originally approved $d$ got approvals
  for the candidates selected in the first $n$ interations). If at
  this point $d$ is selected, then in the following $2n$ iterations
  the remaining set candidates are chosen and the winning committee
  does not change. This means that $p$ is selected. However, this is
  only possible if $x = 0$, i.e., if the set candidates chosen in the
  first $n$ iterations correspond to an exact cover of $U$.
\end{proof}

A very similar proof also works for the case of GreedyPAV. The main
difference is that now including a candidate in a committee does not
allow us to forget about all the voters that approve him or her. The
proof is in \Cref{app:greedy-pav}.

\begin{theorem}\label{thm:greedy-pav}
  GreedyPAV-\textsc{Add-Robustness-Radius} is $\np$-complete.
\end{theorem}

The proof for the case of Phragm{\'e}n-\textsc{Add-Robustness-Radius}
is similar in spirit to the preceding two, but requires careful
calculation of the times when particular groups of voters can purchase
respective candidates.

\begin{theorem}\label{thm:phr}
  Phragm{\'e}n-\textsc{Add-Robustness-Radius} is $\np$-complete.
\end{theorem}
\begin{proof}
  We give a reduction from \textsc{RX3C}. As input, we get a universe
  set $U = \{u_1, \ldots, u_{3n}\}$ and a family
  $\calS = \{S_1, \ldots, S_{3n}\}$ of size-$3$ subsets of $U$. Each
  element of $U$ appears in exactly three sets from $\calS$.
  We ask if
  there is a collection of $n$ sets that form an exact cover of $U$.

  Our reduction proceeds as follows. First, we define two numbers,
  $T = 900n^{12}$ and $t = 30n^5$.  The intuition is that both numbers
  are very large, $T$ is significantly larger than $t^2$, and $t$ is
  divisible by $6n$ (the exact values of $T$ and $t$ are not crucial;
  we did not minimize them but, rather, used values that clearly work
  and simplify the reduction). We form an election $E = (C,V)$ with
  candidate set $C = \{S_1, \ldots, S_{3n}\} \cup \{p,d\}$ and the
  following voters:
  \begin{enumerate}
  \item For each $S_i \in \calS$, there are $T$ voters that approve
    candidate $S_i$. We refer to them as the $\calS$-voters.
  \item For each $u_\ell \in U$, there are $t^2$ voters that approve
    those candidates $S_i$ that correspond to the sets containing
    $u_\ell$. We refer to them as the universe voters.  For each
    $u_\ell \in U$, $\frac{t}{3n}$ of $u_\ell$'s universe voters
    additionally approve candiate $d$.  We refer to them as the
    $d$-universe voters.
    
  \item There are $T+3t^2-2t$ voters that approve both $p$ and $d$. We refer to
    them as the $p$/$d$-voters.
  \item There are $\frac{t}{6n}$ voters that approve
    $p$. We refer to them as the $p$-voters.
  \item There are $n$ voters who do not approve any candidate, and to
    whom we refer as the empty voters.
  \end{enumerate}
  The committee size is $k = 3n+1$ and we can add up to $B = n$
  approvals. The tie-breaking order is:
  \[
    S_1 \pref S_2 \pref \cdots \pref S_{3n} \pref d \pref p.
  \]
  In this election, each candidate $S_i$ is approved by exactly
  $T+3t^2$ voters, $d$ is approved by $(T+3t^2-2t) + t$ voters, and
  $p$ is approved by $(T+3t^2-2t) +\frac{t}{6n}$
  voters.

  \begin{figure}
    \centering
    \begin{tikzpicture}[scale=1]
      \draw[->] (0,0) -- (10,0);
      \foreach \x\l in {1/$0$,3/$A$,5/$B_{pd}$,7/$C$,9/$D$}{
        \draw (\x,0.1) -- (\x,-0.1) node[anchor=north] {\l};
        
      }
      \draw[decorate, decoration = {brace, amplitude=5pt}] (7,0.5) -- (9,0.5) node[pos=0.5,above=16pt,black]{{\scriptsize remaining set candidates}} node[pos=0.5,above=7pt,black]{{\scriptsize are selected\phantom{y}}};      

      \draw[decorate, decoration = {brace, amplitude=5pt}] (1,0.5) -- (3,0.5) node[pos=0.5,above=16pt,black]{{\scriptsize no candidates}} node[pos=0.5,above=7pt,black]{{\scriptsize selected yet}};      

      \draw[<-] (3,-0.5) -- (3,-0.88) node[anchor=north] {{\scriptsize first up to $n$ set}} node[below=7pt]{{\scriptsize candidates are selected}};
      \draw[<-] (5,0.2) -- (5,0.7)  node[anchor=south,above=7pt] {{\scriptsize $p$ or $d$ is selected}} node[anchor=south,above=0pt] {{\scriptsize by this time}};

      \draw[decorate, decoration = {brace, amplitude=5pt,mirror,aspect=0.85}] (3.2,-0.55) -- (6.8,-0.55) node[pos=0.95,below=7pt,black]{{\scriptsize no set candidates are}} node[pos=0.95,below=16pt,black]{{\scriptsize selected during this time}};     
      
    \end{tikzpicture}
    
    \caption{Timeline for the Phragm{\'e}n rule acting    
      on the election from Theorem~\ref{thm:phr}.}
    \label{fig:phr}
  \end{figure}
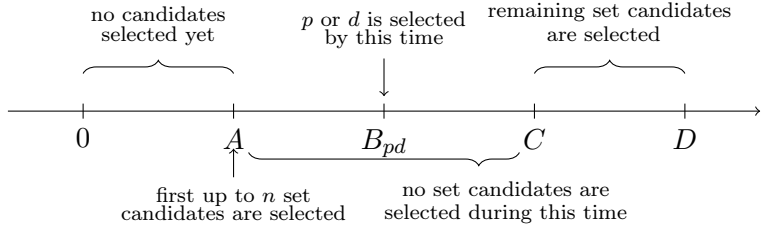

  Let us consider how Phragm{\'e}n operates on this election (we
  encourage the reader to consult Figure~\ref{fig:phr} while reading
  the following text). First, we observe that when we reach time point
  $D = \frac{1}{T}$ then all the not-yet-selected set candidates (for
  whom there still is room in the committee) are selected. Indeed, at
  time $D$ the $\calS$-voters collect enough funds to buy them. On the
  other hand, the earliest point of time when some voters can afford
  to buy a candidate is $A = \frac{1}{T+3t^2}$. Specifically, at time
  $A$ set voters and universe voters jointly purchase up to $n$ set
  candidates (selected sequentially, using the tie-breaking order and
  taking into account that when some candidate is purchased then all
  his or her voters spend all their so-far collected money). Let us
  consider some candidate $S_i$ that was not selected at time point
  $A$. Since $S_i$ was not chosen at $A$, at least $t^2$ of the $3t^2$
  universe voters that approve $S_i$ paid for another candidate at
  time $A$. Thus the earliest time when voters approving $S_i$ might
  have enough money to purchase him or her is $C$, such that:
  \[
    \underbrace{C(T+2t^2)}_{\substack{\text{money earned by those voters} \\ \text{who did not spend it at time $A$}}} +
    \underbrace{(C-A)t^2}_{\substack{\text{money earned between times $C$ and $A$ by uni-} \\ \text{verse voters who paid for candidates at time $A$}}}
    = 1.
  \]
  Simple calculations show that $C = \frac{1+At^2}{T+3t^2}$. Noting
  that $A = \frac{1}{T+3t^2}$, we have that $C = A + A^2t^2$.
  However, prior to reaching time point $C$, either candidate $p$ or
  candidate $d$ is selected. Indeed, at time point
  $B_{pd} = \frac{1}{T+3t^2-2t}$ the $p$/$d$-voters alone would have
  enough money to buy one of their candidates:  We show that $B_{pd} <
  C$, or, equivalently,  that $\frac{1}{B_{pd}} > \frac{1}{C}$.  It
  holds that $\frac{1}{B_{pd}} = T+3t^2 - 2t$ and:
  \[
    \frac{1}{C} = \frac{1}{A+A^2t^2} = \frac{1}{A} \cdot \frac{1}{1+At^2}
    = \frac{T+3t^2}{1+\frac{t^2}{T+3t^2}}
    = \frac{(T+3t^2)^2}{T+4t^2}.
  \]
  By simple transformations, $\frac{1}{B_{pd}} > \frac{1}{C}$ is equivalent to:
  \[
    (T+3t^2 - 2t)(T+4t^2) > (T+3t^2)^2.
  \]
  The left-hand side of this inequality can be expressed as:   
  \[
    ((T+3t^2) - 2t)((T+3t^2) + t^2) = (T+3t^2)^2 +
    \underbrace{(t^2-2t)(T+3t^2) - 2t^3}_{\substack{\text{positive because $t^2-2t > 2t$} \\ \text{due to our assumptions} }},
  \]
  and, hence, our inequality holds. All in all, we have
  $A < B_{pd} < C < D$.

  It remains to consider which among $p$ and $d$ is selected.  If $p$
  were to be selected, then it would happen at time point
  $B_p = \frac{1}{T^2+3t^2-2t + \frac{t}{6n}}$. This is when the
  $p$/$d$-voters and $p$-voters would collect enough money to purchase $p$
  (assuming the former would not spend it on $d$ earlier).  Now, if at
  time $A$ fewer than $n$ set candidates were selected, then at least
  $\frac{t}{3n}$ of the $d$-universe voters would retain their money
  and, hence, $d$ would be selected no later than at time point
  $B_d = \frac{1}{T^2+3t^2-2t+\frac{t}{3n}} < B_p$.  On the other
  hand, if at time point $A$ exactly $n$ set candidates were selected
  (who, thus, would have to correspond to an exact cover of $U$) then all
  the $d$-universe voters would lose their money and voters who
  approve~$d$ would not have enough money to buy him or her before
  time $B_p$.  Indeed, in this case the money accumulated by voters
  approving $d$ would at time $B_p$ be:
  \[
    X = \underbrace{\frac{T+3t^2-2t}{T+3t^2-2t +
        \frac{t}{6n}}}_{\text{money of the $p$/$d$ voters}} +
    \underbrace{t \left( \frac{1}{T+3t^2-2t+\frac{t}{6n}} -
        \frac{1}{T+3t^2}\right)}_{\substack{\text{money collected by
          the $d$-universe} \\ \text{voters between time points $A$ and
          $B_p$}}}
  \]
  We claim that $X < 1$, which is equivalent to the following
  inequality (where we replace $T+3t^2$ with $M$; note that
  $M = \frac{1}{A}$):
  \begin{align*}
    \frac{M-t}{M-2t+\frac{t}{6n}} < 1 + \frac{t}{M} = \frac{M+t}{M} 
  \end{align*}
  By simple transformations, this inequality is equivalent to:
  \begin{align*}
  0 < \frac{Mt+t^2}{6n} - 2t^2, 
  \end{align*}
  which holds as $t > 6n$ and $M > 2t^2$.  To conclude, if at time
  point $A$ there are $n$ set candidates selected for the committee,
  then $p$ is selected for the committee as well.

  Finally, we observe that irrespective of which among $p$ and $d$ is
  selected for the committee, the voters that approve the other one do
  not collect enough money to buy him or her until time $D$.  Thus the
  winning committee either consists of all the set candidates and $d$,
  or of all the set candidates and $p$, where the latter happens
  exactly if at time $A$ candidates corresponding to an exact cover of
  $U$ are selected.\medskip

  If Phragm{\'e}n would choose candidates corresponding to an exact
  cover of $U$ at point $A$, then our reduction outputs a fixed
  yes-instance (as we have just found that an exact cover
  exists). Otherwise, we output the just-described election with
  committee size $k = 3n+1$ and ability to add $B = n$ approvals.  To
  see why this reduction is correct, we make the following three
  observations:
  \begin{enumerate}
  \item By adding $n$ approvals, we cannot significantly modify any of
    the time points $A$, $B_d$, $B_p$, $B_{pd}$, $C$, and $D$ from the preceding
    analysis, except that we can ensure which (up to) $n$ sets are
    first considered for inclusion in the committee just before time
    point $A$.
  \item If there is a collection of $n$ sets in $\calS$ that form an exact
    cover of $U$, then---by the above observation---we can ensure that these
    sets are selected just before time point $A$ (by adding one approval for each
    of them among $n$ distinct empty voters). Hence, if there is an exact cover
    then---by the preceding discussions---we can ensure that the winning committee changes
    (to consist of all the set candidates and $p$).
  \item If there is no exact cover of $U$, then no matter where we add
    (up to) $n$ approvals, candidate $d$ gets selected and, so, the
    winning committee does not change (in particular, even if we add
    $n$ approvals for $p$).
  \end{enumerate}
  Since the reduction clearly runs in polynomial time, the proof is complete.
\end{proof}

It remains to argue that \textsc{Remove-Robustness-Radius} and
\textsc{Swap-Robustness-Radius} also are $\np$-complete for each of
our rules. This, however, is easy to see. In each of the three proofs
above, we can add up to $B = n$ approvals and there are $n$ voters
with empty approval sets.  We were using these $n$ voters to add a
single approval for each of the $n$ sets forming an exact cover,
leading to the selection of $p$ instead of $d$. For the case of
removing approvals, it suffices to replace the $n$ empty voters with
$3n$ ones, such that each set candidate is approved by exactly one of
them, and to allow removeing up to $B = 2n$ approvals. Now we can
achieve the same result as before by deleting approvals for those set
candidates that do not form an exact cover. On the other hand, for
\textsc{Swap-Robustness-Radius}, we add $n$ dummy candidates and we
let each of the empty voters approve a unique dummy candidate. We
allow $n$ approvals to be moved. Then, we can achieve the same effect
as with adding approvals (to the empty voters), by moving approvals
from the dummy candidates to the set ones. Hence, the following holds.
Importantly, the operations that we can perform suffices to change the
order in which the set candidates are considered, but are too small to
affect the general logic of how our rules operate on the constructed
elections.

\begin{corollary}
  Let $\calR$ be one of GreedyCC, GreedyPAV, and
  Phragm{\'e}n. $\calR$-\textsc{Remove-Robustness-Radius} and
  $\calR$-\textsc{Swap-Robustness-Radius} are $\np$-complete.
\end{corollary}

\section{Conclusion and Further Work}

We have adapted the robustness framework of Bredereck et
al.~\cite{bre-fal-kac-nie-sko-tal:j:robustness} to the approval
setting. We have shown that the robustness levels of our rules are
either $1$ or~$k$ (so small perturbations of the votes are either
guaranteed to have minor effect only, or may completely change the
results), and we have studied the complexity of deciding if a given
number of perturbations can change election results (as well as the
complexity of counting the number of ways to achieve this effect).

\paragraph{To Arkadii!}\quad Piotr Faliszewski is very grateful
to Arkadii Slinko for all the work they have done together (and all
that is still coming up). Arkadii opened the door to the studies of
multiwinner voting for Piotr (and many others!), one of their first
joint papers was on the complexity of swap
bribery~\cite{elk-fal-sli:c:swap-bribery}, a problem very similar to
\textsc{Robustness Radius}, so connecting both multiwinner voting and
bribery-style problems in honor of Arkadii seems quite fitting.

\paragraph{Acknowledgements.}\quad
The research presented in this paper has been partially supported by
AGH University of Science and Technology and the ``Doktorat
Wdrożeniowy'' program of the Polish Ministry of Science and Higher
Education, as well as the funds assigned by Polish Ministry of Science
and Technology to AGH University. This project has received funding
from the European Research Council (ERC) under the European Union’s
Horizon 2020 research and innovation programme (grant agreement No
101002854).
\begin{center}
\noindent \includegraphics[width=2.5cm]{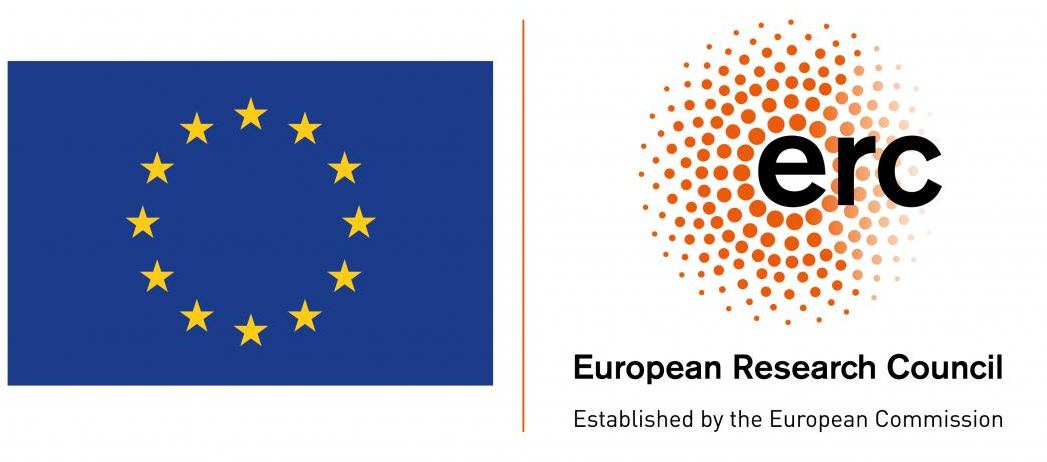}
\end{center}

\bibliography{bib-robust}

\appendix

\section{Proof of  \Cref{thm:sav-remove-counting}}\label{app:sav-remove-counting}

$\sharpp$-completeness follows in the same way as in the proof of
\Cref{thm:sav-add-counting}. For $\sharpp$-hardness, we give a
reduction from the \textsc{\#Perfect-Matchings} problem. Our
construction is the same as in the proof of
Theorem~\ref{thm:sav-add-counting} except that we ensure that every
voter approves of $n^2$ candidates (so each filler candidate approves
one vertex candidate and $n^2-1$ dummy candidates, whereas each edge
voter approves $2$ vertex candidates and $n^2-2$ dummy ones). We
maintain that initially each vertex candidate has score $2$ (so we add
many more filler voters) and each filler candidate has score
$\frac{1}{n^2}$ (the latter score is a consequence of the fact that
each dummy candidate is approved by exactly one voter).  We use the
same notation as in the proof of Theorem~\ref{thm:sav-add-counting}.

Let us now consider how the scores of the vertex candidate change as we delete $n$
approvals and, in consequence, whether the set of winning committees changes or not.
We consider several cases:
\begin{description}
\item[Some vertex candidate loses an approval.] If we delete an
  approval for some vertex candidate~$x$, then his or her score
  decreases by $\frac{1}{n^2}$ due to this. The only way to increase
  $x$'s score is to remove some other approvals in the votes where $x$
  is approved. One can verify that the largest increase possible with
  deleting $n-1$ approvals is to remove them in the same vote; this
  increases $x$'s score by
  $\frac{1}{n^2-n-1}-\frac{1}{n^2} =
  \frac{n-1}{n^2(n^2-n+1}$. However,
  $\frac{-1}{n^2} + \frac{n-1}{n^2(n^2-n+1)} < 0$ so if we delete some
  approval for $x$, his or her score certainly becomes lower than
  $2$. On the other hand, by deleting $n$ approvals it is possible to
  decrease scores of at most $n$ out of $2n$ vertex candidate.  Thus,
  if we delete an approval for some vertex candidate, then not all
  vertex candidates will have equal score and the set of winning
  committees will change.

\item[Some dummy candidate loses an approval in a filler vote.] If we delete an approval for
some dummy candidate in a filler vote, then the score of the vertex candidate $x$ approved by this
voter increases (and from the previous case we know that it is impossible to bring $x$'s score 
down to exactly $2$ by deleting some of his approvals). By deleting further $n-1$ approvals, it
is possible, at best, to increase the scores of $2(n-1)$ further vertex candidates. As a consequence,
there certainly is some vertex candidate whose score remains $2$ and, thus, the set of winning
committees certainly changes.

\item[Approvals are deleted only from dummy candidates in the edge
  voters.]  Let us consider the case where we delete approvals only
  from the dummy candidates in edge voters.  Using a simple counting
  argument (analogous to the ones used previously), one can verify
  that if we delete more than one approval in a single edge vote or if
  we delete approvals from edge votes that do not form a perfect
  matching then some vertex candidates end up with higher score than
  others. In consequence, the set of winning committees changes. On
  the other hand, if we delete a single dummy candidate from each of
  the $n$ edge votes that correspond to a perfect matching then the
  scores of all the vertex candidates increase by the same ammount and
  the set of winning committees does not change
\end{description}

As in the proof of Theorem~\ref{thm:sav-add-counting} it remains to note that if $M$ is the
number of perfect matchings in our input graph then there are exactly $M \cdot (n^2-2)^n$ ways
of deleting $n$ approvals without changing the set of winning committees (for each of the $n$
edges in each of the perfect matchings there are $n^2-2$ independent choices for removing
a single dummy candidate).

\section{Proof of Theorem~\ref{thm:greedy-pav}}\label{app:greedy-pav}

  The proof is analogous to that of Theorem~\ref{thm:greedy-cc} and we
  only modify some details of the argument. We reduce from
  \textsc{RX3C} and the input instance consists of the universe set
  $U = \{u_1, \ldots, u_{3n}\}$ and family
  $\calS = \{S_1, \ldots, S_{3n}\}$ of three-element subsets of
  $U$. Each member of $U$ belongs to exactly three sets from~$\calS$.

  We have two integers, $T = 10n^5$ and $t = 10n^3$, interpreted as
  two large numbers, with $T$ significantly larger than $t$.  We form
  an election $E = (C,V)$ with candidate set
  $C = \{S_1, \ldots, S_{3n}\} \cup \{p,d\}$, and with the following
  voters:
  \begin{enumerate}
  \item For each $S_i \in \calS$, there are $T$ voters that approve
    candidate $S_i$.
  \item For each two sets $S_i$ and $S_{j}$, there are $T$ voters
    that approve candidates $S_i$ and $S_{j}$.
  \item There are $2nT+ 0.5nT+4nt$ voters that approve $p$ and $d$.
  \item For each $u_\ell \in U$, there are $t$ voters that approve $d$
    and those candidates $S_i$ that correspond to the sets containing
    $u_\ell$.
  \item There are $1.5nt$ voters who approve $p$.
  \item There are $n$ voters who do not approve any candidates.
  \end{enumerate}
  The committee size is $k = 3n+1$ and we can add up to $B = n$
  approvals.  We assume that the tie-breaking order among the
  candidates is:
  \[
    S_1 \pref S_2 \pref \cdots \pref S_{3n} \pref p \pref d.
  \]
  Prior to adding any approvals and running the rule, each candidate $S_i$ has score
  $3nT+3t$, $p$ has score $2nT + 0.5nT +4nt + 1.5nt$, and $d$ has score $2nT+0.5nT+4nt + 3nt$.

  Let us now consider how GreedyPAV operates on this election. Prior
  to the first iteration, all the set candidates have the same score,
  much higher than that of $p$ and $d$. Due to the tie-breaking order,
  GreedyPAV chooses $S_1$. As a consequence, the scores of all other
  set candidates decrease to $(3n-1)T + 0.5T$ plus some number of
  points from the fourth group of voters (but this part of the score
  is much smaller than that from the first two groups of
  voters). GreedyPAV acts analogously during the first $n$ iterations
  and it chooses a family $\calT$ of $n$ set elements.
  
  After the first $n$ iterations, each of the remaining $2n$ set
  candidates has $2nT + 0.5nT$ points from the first two groups of
  voters and at most $3nt$ points from the fourth group (in fact less,
  but this bound suffices). On the other hand, both $p$ and $d$ have
  at least $2nT + 0.5nT + 4nt$ points and, so, in the next
  iteration the algorithm chooses one of them.
  Specifically, $p$ has score $2nT+0.5nT+4nt+1.5nt$ and $d$ has score
  $2nT+0.5nT+4nt + 1.5nt + x$, where the value of $x$ is as
  follows. If~$\calT$ corresponds to an exact cover of $U$, then $x$
  is $0$ because for each set candidate that GreedyPAV adds during the
  first $n$ iterations, $d$ loses $1.5t$ points from the fourth group
  of voters.  However, if $\calT$ does not correspond to an exact
  cover of $U$, then for at least one added set candidate the score of
  $d$ does not decrease by $1.5t$ but by at most
  $t+(\nicefrac{1}{2}-\nicefrac{1}{3})t$. Thus $x$ is at least
  $\frac{t}{3}$. So, if $\calT$ corresponds to an exact cover of $U$
  then $p$ is selected and otherwise $d$ is.  In either case, the
  score of the unselected one drops by at least $1.25nT$, which means
  that in the following $2n$ iterations the remaining set candidates
  are selected (because each of them has score at least $1.5nT$).
  
  If candidate $p$ is selected without adding any approvals, then it
  means that we can find a solution for the \textsc{RX3C} instnace
  using a simple greedy algorithm. In this case, instead of outputting
  the just-described instance of
  GreedyPAV-\textsc{Add-Robustness-Radius}, we output a fixed one, for
  which the answer is \emph{yes}. Otherwise, we know that without
  adding any approvals the winning committee is
  $\{S_1, \ldots, S_{3n},d\}$. We focus on this latter case.

  If there is an exact cover of $U$ by sets from $\calS$ then it
  suffices to add a single approval for each of the corresponding sets
  in the last group of voters (those that originally do not approve
  anyone). Then, by the same analysis as above, we can verify that the
  sets forming the cover are selected in the first $n$ iterations,
  followed by $p$, followed by all the other set candidates.
  
  For the other direction, let us note that no matter which $n$
  approvals we add, it is impossible to modify the general scenario
  that GreedyPAV follows: It first chooses $n$ set candidates, then
  either $p$ or $d$ is selected (where the former can happen only if
  the first $n$ set canidates correspond to an exact cover of $U$),
  and finally the remaining $2n$ set candidates are chosen.  This is
  so, because adding $n$ voters results in modifying each of the
  scores by a value between $-n$ and $n$ and such changes do not
  affect our analysis from the preceding paragraphs.  This completes
  the proof.

\end{document}